\documentclass[12pt,letter,sans serif]{article}
\usepackage{tgpagella}
\usepackage[utf8]{inputenc}
\usepackage{geometry}
\usepackage{setspace,newfloat,cite}
\usepackage{amsmath}
\usepackage{amsfonts}
\usepackage{amssymb}
\usepackage{amsthm}
\usepackage{lscape}
\usepackage{graphicx}
\usepackage{floatrow} 
\usepackage{bbm}
\usepackage{bm}
\usepackage{caption}
\usepackage{wrapfig}
\usepackage{fancyhdr}
\usepackage[table,xcdraw]{xcolor}
\usepackage{hyperref}
\usepackage{caption}
\usepackage{booktabs}
\usepackage{graphicx}
\usepackage{dsfont}
\usepackage{csquotes}
\usepackage{zref}
\usepackage[square,sort,comma,longnamesfirst]{natbib}
\usepackage{algorithm, algorithmic}
\usepackage{subcaption}

\definecolor{teal}{HTML}{77C4D3}
\definecolor{newred}{HTML}{E66100}
\definecolor{newgreen}{HTML}{44AA99}
\definecolor{newyellow}{HTML}{E899D5}
\makeatletter
\let\@msm@th@eqref\eqref
\renewcommand{\eqref}[1]{%
	\begingroup
	\leavevmode
	\color{newred}%
	\hypersetup{linkbordercolor=[named]{newred}}%
	\@msm@th@eqref{#1}%
	\endgroup
}
\makeatother

\bibpunct{(}{)}{;}{a}{,}{,}

\makeatother
\hypersetup{colorlinks,linkcolor={newred},citecolor={newgreen},urlcolor={newyellow}} 
\DeclareMathOperator*{\E}{\mathbb{E}}
\newcommand{\argmin}{\arg\!\min}

\date{April 18, 2021 \vspace{.1cm}}
\usepackage{eucal} 
\usepackage[shortlabels]{enumitem}
\def\defeq{\equiv}

\newcommand{\EE}[2]{\mathbb{E}_{#1\!\!}\left[#2\right]}

\def\E#1{\EE{\,}{#1}}
\def\bR{{\mathbf R}}
\def\bI{{\mathbf I}}

\def\bA{{\mathbf A}}
\def\br{{\mathbf r}}
\def\bf{{\mathbf f}}
\def\bg{{\mathbf g}}
\def\bF{{\mathbf F}}
\def\bu{{\mathbf u}}

\def\bw{{\mathbf w}}
\def\bB{{\mathbf B}}

\def\bU{{\mathbf U}}
\def\bV{{\mathbf V}}
\def\bX{{\mathbf X}}
\def\bD{{\mathbf D}}

\def\bT{{\mathbf T}}
\def\bJ{{\mathbf J}}

\def\bE{{\mathbf E}}

\def\bC{{\mathbf C}}
\def\bx{{\mathbf x}}

\def\be{{\mathbf e}}

\def\by{{\mathbf y}}

\def\bgamma{{\bm \gamma}}

\def\bLambda{{\bm \Lambda}}
\def\bOmega{{\bm \Omega}}
\def\biota{{\bm \iota}}
\def\balpha{{\bm \alpha}}

\def\bSigma{{\bm \Sigma}}
\def\bGamma{{\bm \Gamma}}
\def\bTheta{{\bm \Theta}}

\def\boldm{{\mathbf m}}
\def\boldeta{{\bm \eta}}
\def\bvarepsilon{{\bm \varepsilon}}

\usepackage{mathtools}
\usepackage{wasysym}
\usepackage{indentfirst}
\DeclareMathAlphabet{\mathcal}{OMS}{cmsy}{m}{n}

\DeclarePairedDelimiter\abs{\lvert}{\rvert}%
\DeclarePairedDelimiter\norm{\lVert}{\rVert}%

\makeatletter
\let\oldabs\abs
\def\abs{\@ifstar{\oldabs}{\oldabs*}}
\let\oldnorm\norm
\def\norm{\@ifstar{\oldnorm}{\oldnorm*}}
\makeatother

\newcommand{\1}[1]{\mathds{1}\left[#1\right]}

\makeatletter
\newcommand*{\rom}[1]{\expandafter\@slowromancap\romannumeral #1@}
\makeatother
\theoremstyle{plain}
\newtheorem{thm}{Theorem}

\newtheorem{lem}{Lemma}
\newcommand{\vertiii}[1]{{\left\vert\kern-0.25ex\left\vert\kern-0.25ex\left\vert #1 
		\right\vert\kern-0.25ex\right\vert\kern-0.25ex\right\vert}}
{
	\theoremstyle{plain}
	
}

\makeatletter
\def\mathcolor#1#{\@mathcolor{#1}}
\def\@mathcolor#1#2#3{%
	\protect\leavevmode
	\begingroup
	\color#1{#2}#3%
	\endgroup
}
\makeatother

\numberwithin{equation}{section}
\usepackage{array}
\newcommand{\PreserveBackslash}[1]{\let\temp=\\#1\let\\=\temp}
\newcolumntype{C}[1]{>{\PreserveBackslash\centering}p{#1}}
\newcolumntype{R}[1]{>{\PreserveBackslash\raggedleft}p{#1}}
\newcolumntype{L}[1]{>{\PreserveBackslash\raggedright}p{#1}}

\usepackage[toc, page, title]{appendix}
\definecolor{awesome-pink}{HTML}{EF4089}
\definecolor{awesome-coral}{HTML}{FF4B6B}
\usepackage{titlesec}
\titleformat{\section}
{\centering\normalfont\fontsize{15.5}{15}\scshape}{\thesection}{1em}{}

\titleformat{\subsection}
{\normalfont\fontsize{13}{13}\scshape}{\thesubsection}{1em}{}
\usepackage{setspace}
\usepackage{rotating}
\usepackage[font=small,labelfont=bf]{caption}

\begin{document}	
	\title{{\textsc{A Basket Half Full: Sparse Portfolios}}	}
\author{
	Ekaterina Seregina\thanks{\scriptsize{Correspondence to: Department of Economics, University of California, Riverside, CA 92521, USA. \newline E-mail: \texttt{ekaterina.seregina@email.ucr.edu}. \newline This research did not receive any specific grant from funding agencies in the public, commercial, or not-for-profit sectors. Declarations of interest: none. }}}
	\maketitle
	\thispagestyle{empty}
	
	\begin{abstract}

\linespread{1.3} \selectfont  %%%NEW VERSION (SHORT)
\noindent The existing approaches to sparse wealth allocations (1) are limited to low-dimensional setup when the number of assets is less than the sample size; (2) lack theoretical analysis of sparse wealth allocations and their impact on portfolio exposure; (3) are suboptimal due to the bias induced by an $\ell_1$-penalty. We address these shortcomings and develop an approach to construct sparse portfolios in high dimensions. Our contribution is twofold: from the theoretical perspective, we establish the oracle bounds of sparse weight estimators and provide guidance regarding their distribution. From the empirical perspective, we examine the merit of sparse portfolios during different market scenarios. We find that in contrast to non-sparse counterparts, our strategy is robust to recessions and can be used as a hedging vehicle during such times.

		\vskip 2mm
		\noindent \textit{Keywords}: High-dimensionality, Portfolio Optimization, Factor Investing, De-biasing, Post-Lasso, Approximate Factor Model
		\vskip 2mm
		
		\noindent \textit{JEL Classifications}: C13, C55, C58, G11, G17
		
		\newpage
	\end{abstract} 
	
	%%%%%%%%%%%%%%%%%%%%%%%%
	\newpage 
	\setlength{\baselineskip}{20pt}
	\setcounter{page}{1}	
	\section{Introduction}
	The search for the optimal portfolio weights reduces to the questions (i) which stocks to buy and (ii) how much to invest in these stocks. Depending on the strategy used to address the first question, the existing allocation approaches can be further broken down into the ones that invest in all available stocks, and the ones that select a subset out of the stock universe. The latter is referred to as a \textit{sparse portfolio}, since some assets will be excluded and get a zero weight leading to sparse wealth allocations. Any portfolio optimization problem requires the inverse covariance matrix, or \textit{precision} matrix, of excess stock returns as an input. In the era of big data, a search for the optimal portfolio becomes a high-dimensional problem: the number of assets, $p$, is comparable to or greater than the sample size, $T$. Constructing non-sparse portfolios in high dimensions has been the main focus of the existing research on asset management for a long time. In particular, many papers focus on developing an improved covariance or precision estimator to achieve desirable statistical properties of portfolio weights. In contrast, the literature on constructing sparse portfolio is scarce: it is limited to a low-dimensional framework and lacks theoretical analysis of the resulting sparse allocations. In this paper we fill this gap and propose a novel approach to construct sparse portfolios in high dimensions. We obtain the oracle bounds of sparse weight estimators and provide guidance regarding their distribution. From the empirical perspective, we examine the merit of sparse portfolios during the periods of economic growth, moderate market decline and severe economic downturns. We find that in contrast to non-sparse counterparts, our strategy is robust to recessions and can be used as a hedging vehicle during such times.
	
	As pointed out above, estimating high-dimensional covariance or precision matrix to improve portfolio performance of non-sparse strategies has received a lot of attention in the existing literature. \cite{Ledoit2004,Ledoit2017} developed linear and non-linear shrinkage estimators of covariance matrix, \cite{fan2013POET,fan2018elliptical} introduced a covariance matrix estimator when stock returns are driven by common factors under the assumption of a spiked covariance model. Once the covariance estimator is obtained, it is then inverted to get a precision matrix, the main input to any portfolio optimization problem. A parallel stream of literature has focused on estimating precision matrix directly, that is, avoiding the inversion step that leads to additional estimation errors, especially in high dimensions. \cite{GLASSO} developed an iterative algorithm that estimates the entries of precision matrix column-wise using penalized Gaussian log-likelihood (Graphical Lasso); \cite{meinshausen2006} used the relationship between regression coefficients and the entries of precision matrix to estimate the elements of the latter column by column (nodewise regression). \cite{cai2011constrained} use constrained $\ell_1$-minimization for inverse matrix estimation (CLIME). \cite{Caner2019} examined the performance of high-dimensional portfolios constructed using covariance and precision estimators and found that precision-based models outperform covariance-based counterparts in terms of the out-of-sample (OOS) Sharpe Ratio and portfolio return.
	
	From a practical perspective, apart from enjoying favorable statistical properties a successful wealth allocation strategy should be easy to maintain and monitor and it should be robust to economic downturns such that investors could use it as a hedging vehicle. Having this motivation in mind, we chose several popular covariance and precision-based estimators to construct non-sparse portfolios and explore their performance during the recent COVID-19 outbreak. Using daily returns of 495 constituents of the S\&P500 from May 25, 2018 -- September 24, 2020 (588 obs.), Table \ref{table1} reports the performance of the selected strategies: we included equal-weighted (EW) and Index portfolios, as well as precision-based nodewise regression estimator by \cite{meinshausen2006} (motivated by the recent application of this statistical technique to portfolio studied in \cite{Caner2019}), linear shrinkage covariance estimator by \cite{Ledoit2004} and CLIME by \cite{cai2011constrained}. We use May 25, 2018 -- October 23, 2018 (105 obs.) as a training period and October 24, 2018 -- September 24, 2020 (483 obs.) as the out-of-sample test period. We roll the estimation window over the test sample to rebalance the portfolios monthly. The left panel of Table \ref{table1} shows return, risk and Sharpe Ratio of portfolios over the training period, and the right panel reports cumulative excess return (CER) and risk over two sub-periods of interest: before the pandemic (January 2, 2019 -- December 31, 2019) and during the first wave of COVID-19 outbreak in the US (January 2, 2020 -- June 30, 2020). As evidenced by Table \ref{table1}, none of the  portfolios was robust to the downturn brought by pandemic and yielded negative CER. We noticed that similar pattern pertained in several other historic episodes of mild and severe downturns, such as the Global Financial Crisis (GFC) of 2007-09.\footnote{Please see the Empirical Application section for more details.}
	
	Studies that examine the relationship between portfolio performance and the number of stock holdings are scarce. \cite{Tidmore2019} used active US equity funds' quarterly data from January 2000 to December 2017 from Morningstar, Inc. to study the impact of concentration (measured by the number of holdings) on fund excess returns: they found that the effect was significant and fluctuated considerably over time. Notably, the relationship became negative in the period preceding and including the GFC. This indicates that holding sparse portfolios might be the key to hedging during downturns. To support this hypothesis, we further compare the performance of sparse vs non-sparse strategies in terms of utility gain to investors. Suppose we observe $i=1,\ldots,p$ excess returns over $t=1,\ldots,T$ period of time:  $\br_t=(r_{1t},\ldots,r_{pt})' \ {\sim}\ \mathcal{D} (\boldm, \bSigma)$.
	Consider the following mean-variance utility problem: $\text{min}_{\bw} \ -U \defeq \frac{\gamma}{2}\bw\bSigma\bw - \bw'\boldm, \ \text{s.t.} \ \bw'\biota = 1, \ \abs{\text{supp}(\bw)}\leq \bar{p}, \ \bar{p} \leq p$, where $\bw$ is a $p\times 1$ vector of portfolio weights, $\text{supp}(\bw)=\{i:w_i>0\}$ is the cardinality constraint that controls sparsity, and $\gamma$ determines the risk of an investor under the assumption of a normal distribution. When $\bar{p}=p$ the portfolio is non-sparse and the respective utility is denoted as $U^{\text{Non-Sparse}}$, while when $\bar{p}<p$ the utility of such sparse portfolio is denoted as $U^{\text{Sparse}}$. Figure \ref{f00} reports the ratio of utilities using monthly data from 2003:04 to 2009:12 on the constituents of the S\&P100 as a function of $\bar{p}$: we set $\gamma=3$ and vary $\bar{p}=\{5,10,15,20,30,\ldots,90\}$\footnote{Since the optimization problem with a cardinality constraint is not convex, we find a solution using Lagrangian relaxation procedure of \cite{Shaw2008}}. Our test sample includes two periods of particular interest: before the GFC (2004:01-2006:12) and during the GFC (2007:01-2009:12) As evidenced from Figure \ref{f00}: (1) for both time periods there exists a lower-dimensional subset of stocks which brings greater utility compared to non-sparse portfolios; (2) the number of stocks minimizing the ratio of utilities is smaller during the GFC compared to the period preceding it. Both findings are consistent with the empirical result of \cite{Tidmore2019} that including more stocks does not guarantee better performance and suggesting that holding a ``basket half full" instead can help achieve superior performance even in stressed market scenarios.	 
\begin{table}[htb!]
	\centering
	\resizebox{\textwidth}{!}{%
		\begin{tabular}{@{}cccccccc@{}}
			\toprule
			& \multicolumn{3}{c}{\begin{tabular}[c]{@{}c@{}}Total OOS Performance\\ 10/24/19--09/24/20\end{tabular}}            & \multicolumn{2}{c}{\begin{tabular}[c]{@{}c@{}}Before the Pandemic\\ 01/02/19--12/31/19\end{tabular}} & \multicolumn{2}{c}{\begin{tabular}[c]{@{}c@{}}During the Pandemic\\ 01/02/20--06/30/20\end{tabular}} \\ \midrule
			& \textbf{\begin{tabular}[c]{@{}c@{}}Return\\($\times 100$)\end{tabular}} & \textbf{\begin{tabular}[c]{@{}c@{}}Risk\\($\times 100$)\end{tabular}} & \textbf{\begin{tabular}[c]{@{}c@{}}Sharpe Ratio\\ {}\end{tabular}} & \textbf{\begin{tabular}[c]{@{}c@{}}CER\\($\times 100$)\end{tabular}}      & \textbf{\begin{tabular}[c]{@{}c@{}}Risk\\($\times 100$)\end{tabular}}     & \textbf{\begin{tabular}[c]{@{}c@{}}CER\\($\times 100$)\end{tabular}}      & \textbf{\begin{tabular}[c]{@{}c@{}}Risk\\($\times 100$)\end{tabular}}     \\ \midrule
			EW & 0.0108 & 1.8781 & \multicolumn{1}{c|}{0.0058} & 28.5420 & 0.8010 & -19.7207 & 3.3169 \\
			Index & 0.0351 & 1.7064 & \multicolumn{1}{c|}{0.0206} & 27.8629 & 0.7868 & -9.0802 & 2.9272 \\
			Nodewise Regr'n & 0.0322 & 1.6384 & \multicolumn{1}{c|}{0.0196} & 29.6292 & 0.6856 & -11.7431 & 2.8939\\
			CLIME & 0.0793 & 3.1279 & \multicolumn{1}{c|}{0.0373} & 31.5294 & 1.0215 & -25.3004 & 3.8972 \\
			LW & 0.0317 & 1.7190 & \multicolumn{1}{c|}{0.0184} & 29.5513 & 0.7924 & -14.9328 & 3.0115 \\
			Our Post-Lasso-based & 0.1247 & 1.7254 & \multicolumn{1}{c|}{0.0723} & 45.2686 & 1.0386 & 12.4196 & 2.8554 \\
			Our De-biased Estimator & 0.0275 & 0.5231 & \multicolumn{1}{c|}{0.0526} & 23.7629 & 0.4972 & 6.5813 & 0.5572 \\
			\bottomrule
		\end{tabular}%
	}
	\caption{\textit{Performance of non-sparse and sparse portfolios: return ($\times 100$), risk ($\times 100$) and Sharpe Ratio over the training period (left), CER ($\times 100$) and risk ($\times 100$) over two sub-periods (right). Weights are estimated using the standard Global Minimum Variance formula. In-sample: May 25, 2018 -- October 23, 2018 (105 obs.), Out-of-sample (OOS): October 24, 2018 -- September 24, 2020 (483 obs.)}}
	\label{table1}
\end{table}		
		\begin{figure}[!htbp]
		\centering
		\includegraphics[width=0.77\textwidth]{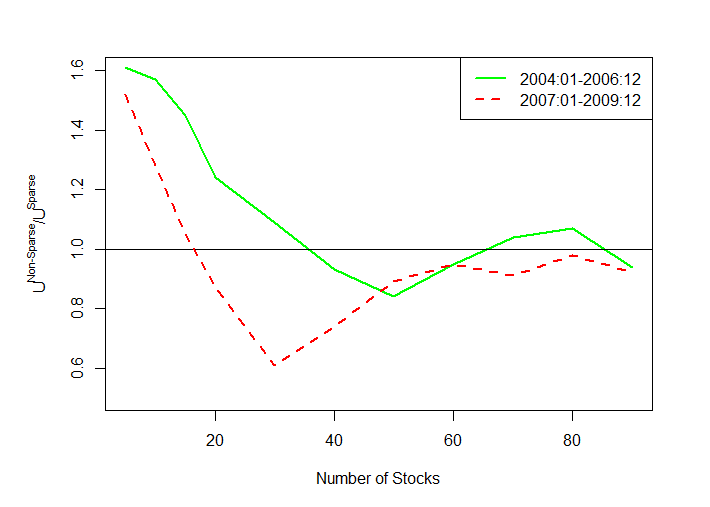}
		\caption{\linespread{1.3} \selectfont\textit{The ratio of non-sparse and sparse portfolio utilities averaged over the test window.}}
		\label{f00}
	\end{figure}

	In order to create a sparse portfolio, that is, a portfolio with many zero entries in the weight vector, we can use an $\ell_1$-penalty (Lasso) on the portfolio weights which shrinks some of them to zero (see \cite{Fan2019Sparse},  \cite{Ao2019}, \cite{Li2015sparse}, \cite{Brodie2009sparse} among others). \cite{caccioli2016liquidity} proved the mathematical equivalence of adding an $\ell_1$-penalty and  
	controlling transaction costs associated with the bid-ask spread impact of single and sequential trades executed in a very short time. This indicates another advantage of sparse portfolios: market liquidity dries up during economic downturns which increases bid-ask spreads, a measure of liquidity costs. Henceforth, regularizing portfolio positions accounts for the increased liquidity risk associated with acquiring and liquidating positions.
	The existing literature on sparse wealth allocations is scarce and has several drawbacks: (1) it is limited to low-dimensional setup when $p<T$, whereas sparsity becomes especially important in high-dimensional scenarios; (2) it lacks theoretical analysis of sparse wealth allocations and their impact on portfolio exposure; (3) the use of an $\ell_1$-penalty produces biased estimates (see \cite{Zhang2014,Javanmard2014confidence,Javanmard2014hypothesis,Buhlmann2014,Belloni2015uniform,Javanmard2018debiasing} among others), however, this issue has been overlooked in the context of portfolio allocation. This paper addresses the aforementioned drawbacks and develops an approach to construct sparse portfolios in high dimensions. Our contribution is twofold: from the theoretical perspective, we establish the oracle bounds of sparse weight estimators and provide guidance regarding their distribution. From the empirical perspective, we examine the merit of sparse portfolios during different market scenarios. We find that in contrast to non-sparse counterparts, our strategy is robust to recessions and can be used as a hedging vehicle during such times. To illustrate, the last two rows of Table \ref{table1} show the performance of two sparse strategies proposed in this paper: both approaches outperform non-sparse counterparts in terms of total OOS Sharpe Ratio, and they produce positive CER during the pandemic, as well as in the period preceding it. Figure \ref{network} shows the stocks selected by post-Lasso in August, 2019 and in May, 2020: the colors serve as a visual guide to identify groups of closely-related stocks (stocks of the same color do not necessarily correspond to the same sector). 
	Our framework makes use of the tool from the network theory called nodewise regression which not only satisfies desirable statistical properties, but also allows us to study whether certain industries could serve as safe havens during recessions. We find that such non-cyclical industries as consumer staples, healthcare, retail and food were driving the returns of the sparse portfolios during both GFC and COVID-19 outbreak, whereas insurance sector was the least attractive investment in both periods.

	This paper is organized as follows: Section 2 introduces sparse de-biased portfolio and sparse portfolio using post-Lasso. Section 3 develops a new high-dimensional precision estimator called Factor Nodewise regression. Section 4 develops a framework for factor investing. Section 5 contains theoretical results and Section 6 validates these results using simulations. Section 7 provides empirical application. Section 8 concludes.
	 	 \phantomsection
	 \subsection*{Notation}
	 \addcontentsline{toc}{section}{Notation}
	 For the convenience of the reader, we summarize the notation to be used throughout the paper. Let $\mathcal{S}_p$ denote the set of all $p \times p$ symmetric matrices. For any matrix $\bC$, its $(i,j)$-th element is denoted as $c_{ij}$.
	Given a vector $\bu\in \mathbb{R}^d$ and parameter $a\in \lbrack1,\infty)$, let $\norm{\bu}_a$ denote $\ell_a$-norm. Given a matrix $\bU \in\mathcal{S}_p$, let $\Lambda_{\text{max}}(\bU) \defeq \Lambda_1(\bU) \geq \Lambda_2(\bU)\geq \ldots \Lambda_{\text{min}}(\bU) \defeq \Lambda_p(\bU)$ be the eigenvalues of $\bU$, and $\text{eig}_K(\bU) \in \mathbb{R}^{K\times p}$ denote the first $K\leq p$ normalized eigenvectors corresponding to $\Lambda_1(\bU), \ldots \Lambda_K(\bU)$. Given parameters  $a,b\in \lbrack1,\infty)$, let $\vertiii{\bU}_{a,b}=\max_{\norm{\by}_a=1}\norm{\bU\by}_{b}$ denote the induced matrix-operator norm. The special cases are $\vertiii{\bU}_1\defeq \max_{1\leq j\leq p}\sum_{i=1}^{p}\abs{u_{i,j}}$ for the $\ell_1/\ell_1$-operator norm; the operator norm ($\ell_2$-matrix norm) $\vertiii{\bU}_{2}^{2}\defeq\Lambda_{\text{max}}(\bU\bU')$ is equal to the maximal singular value of $\bU$; $\vertiii{\bU}_{\infty}\defeq \max_{1\leq j\leq p}\sum_{i=1}^{p}\abs{u_{j,i}}$ for the $\ell_{\infty}/\ell_{\infty}$-operator norm. Finally, $\norm{\bU}_{\text{max}}=\max_{i,j}\abs{u_{i,j}}$ denotes the element-wise maximum, and $\vertiii{\bU}_{F}^{2}=\sum_{i,j}u_{i,j}^{2}$ denotes the Frobenius matrix norm. We also use the following notations: $a\vee b=\max\{a,b\}$, and $a\wedge b=\min\{a,b\}$. For an event $A$, we say that $A \  \text{wp} \rightarrow 1$ when $A$ occurs with probability approaching $1$ as $T$ increases.
	\section{Sparse Portfolios}

There exist several widely used portfolio weight formulations depending on the type of optimization problem solved by an investor. Suppose we observe $p$ assets (indexed by $i$) over $T$ period of time (indexed by $t$). Let $\br_t=(r_{1t}, r_{2t},\ldots,r_{pt})' \sim  \mathcal{D} (\boldm, \bSigma)$ be a $p \times 1$ vector of \textit{excess} returns drawn from a distribution $\mathcal{D}$, where $\boldm$ and $\bSigma$ are unconditional mean and covariance of excess returns, and $\mathcal{D}$ belongs to either sub-Gaussian or elliptical families. When $\mathcal{D} = \mathcal{N}$, the precision matrix $\bSigma^{-1}\defeq \bTheta$ contains information about conditional dependence between the variables. For instance, if $\theta_{ij}$, which is the $ij$-th element of the precision matrix, is zero, then the variables $i$ and $j$ are conditionally independent, given the other variables. The goal of the Markowitz theory is to choose assets weights in a portfolio \textit{optimally}. We will study two criteria of optimality: the first is a well-known Markowitz weight-constrained optimization problem, and the second formulation relaxes constraints on portfolio weights.

 The first optimization problem, which will be referred to as \textit{Markowitz weight-constrained problem (MWC)}, searches for assets weights such that the portfolio achieves a desired expected rate of return with minimum risk, under the restriction that all weights sum up to one.
The aforementioned goal can be formulated as the following quadratic optimization problem:
\begin{equation} \label{sys1}
	\min_{\bw}\frac{1}{2} \bw'\bSigma\bw,  \ \text{s.t.} \ \bw'\biota =1 \ \text{and} \ \boldm'\bw\geq\mu, 
\end{equation}
where $\bw$ is a $p \times 1$ vector of assets weights in the portfolio, $\biota$ is a $p \times 1$ vector of ones, and $\mu$ is a desired expected rate of portfolio return. The constraint in \eqref{sys1} requires portfolio weights to sum up to one - this assumption can be easily relaxed and we will demonstrate the implications of this constraint on portfolio weights.

If $\boldm'\bw>\mu$, then the solution to \eqref{sys1} yields the \textit{global minimum-variance (GMV)} portfolio weights $\bw_{GMV}$:
\begin{equation} \label{eq2}
\bw_{GMV}=(\biota'\bTheta\biota)^{-1}\bTheta\biota.
\end{equation}

If $\boldm'\bw=\mu$, the solution to \eqref{sys1} is
\begin{align}
&\bw_{MWC}=(1-a_1)\bw_{GMV}+a_1 \bw_{M},\label{eq3}\\
&\bw_{M}=(\biota'\bTheta\boldm)^{-1}\bTheta\boldm,  \label{eq4}\\
&a_1=\dfrac{\mu(\boldm'\bTheta\biota)(\biota'\bTheta\biota)-(\boldm'\bTheta\biota)^2}{(\boldm'\bTheta\boldm)(\biota'\bTheta\biota)-(\boldm'\bTheta\biota)^2}, \label{eq5}
\end{align}
where $\bw_{MWC}$ denotes the portfolio allocation with the constraint that the weights need to sum up to one and $\bw_{M}$ captures all mean-related market information.

	The second optimization problem, which will be referred to as \textit{Markowitz risk-constrained (MRC)} problem, has the same objective as in \eqref{sys1}, but portfolio weights are not required to sum up to one:	
\begin{equation} \label{sys2}
\min_{\bw}\frac{1}{2} \bw'\bSigma\bw \quad \text{s.t.} \ \boldm'\bw\geq\mu.
\end{equation}
It can be easily shown that the solution to \eqref{sys2} is:
\begin{equation}\label{w1}
\bw^{*}_{1}=\frac{\mu\bTheta\boldm}{\boldm'\bTheta\boldm}.
\end{equation}
Alternatively, instead of searching for a portfolio with a specified desired expected rate of return and minimum risk, one can maximize expected portfolio return given a maximum risk-tolerance level:
\begin{equation} \label{sys3} 
\max_{\bw} \bw'\boldm  \quad  \text{s.t.} \ \bw'\bSigma\bw\leq \sigma^2.
\end{equation}
In this case, the solution to \eqref{sys3} yields:
\begin{equation}\label{w2}
\bw^{*}_{2}= \frac{\sigma^2}{\bw'\boldm}\bTheta\boldm= \frac{\sigma^2}{\mu}\bTheta\boldm.
\end{equation}
To get the second equality in \eqref{w2} we used the definition of $\mu$ from \eqref{sys1} and \eqref{sys2}. It follows that if $\mu=\sigma\sqrt{\theta}$, where $\theta\defeq \boldm'\bTheta\boldm$ is the squared Sharpe Ratio, then the solution to \eqref{sys2} and \eqref{sys3} admits the following expression:
\begin{equation}\label{ee20}
\bw_{MRC}=\frac{\sigma}{\sqrt{\boldm'\bTheta\boldm}}\Theta\boldm=\frac{\sigma}{\sqrt{\theta}}\balpha,
\end{equation}
where  $\balpha\defeq \bTheta\boldm$. Equation \eqref{ee20} tells us that once an investor specifies the desired return, $\mu$, and maximum risk-tolerance level, $\sigma$, this pins down the Sharpe Ratio of the portfolio which makes the optimization problems of minimizing risk and maximizing expected return of the portfolio in \eqref{sys2} and \eqref{sys3} identical.

This brings us to three alternative portfolio allocations commonly used in the existing literature: Global Minimum-Variance Portfolio in \eqref{eq2}, weight-constrained Markowitz Mean-Variance in \eqref{eq3} and maximum-risk-constrained Markowitz Mean-Variance in \eqref{ee20}. Below we summarize the aforementioned portfolio weight expressions:
\begin{align}
&\text{GMV:} \quad &&\bw_{GMV}=(\biota'\bTheta\biota)^{-1}\bTheta\biota,\label{e156}\\
&\text{MWC} \quad &&\bw_{MWC}=(1-a_1)\bw_{GMV}+a_1 \bw_{M},\label{e157}\\
&\text{where} \quad &&\bw_{M}=(\biota'\bTheta\boldm)^{-1}\bTheta\boldm, \nonumber\\
& &&a_1=\dfrac{\mu(\boldm'\bTheta\biota)(\biota'\bTheta\biota)-(\boldm'\bTheta\biota)^2}{(\boldm'\bTheta\boldm)(\biota'\bTheta\biota)-(\boldm'\bTheta\biota)^2}, \nonumber\\
&\text{MRC:} \quad &&\bw_{MRC}=\frac{\sigma}{\sqrt{\theta}}\balpha, \label{e158}\\
&\text{where} \quad 
&&\balpha =\bTheta\boldm, \quad \theta = \boldm'\bTheta\boldm \nonumber
\end{align}
	
So far we have considered allocation strategies that put non-zero weights to all assets in the financial portfolio. As an implication, an investor needs to buy a certain amount of each security even if there are a lot of small weights. However, oftentimes investors are interested in managing a few assets which significantly reduces monitoring and transaction costs and was shown to outperform equal weighted and index portfolios in terms of the Sharpe Ratio and cumulative return (see \cite{Fan2019Sparse},  \cite{Ao2019}, \cite{Li2015sparse}, \cite{Brodie2009sparse} among others). This strategy is based on holding a \textit{sparse portfolio}, that is, a portfolio with many zero entries in the weight vector. 

\subsection{Sparse De-Biased Portfolio}
Let us first introduce some notations. The sample mean and sample covariance matrix have standard formulas: $ \widehat{\boldm}=\dfrac{1}{T}\sum_{t=1}^{T}\br_{t}$ and $\widehat{\bSigma}=\dfrac{1}{T}\sum_{t=1}^{T}(\br_t-\widehat{\boldm})(\br_t-\widehat{\boldm})^{'}$. Our empirical application shows that risk-constrained Markowitz allocation in \eqref{e158} outperforms GMV and MWC portfolios in \eqref{e156}-\eqref{e157}. Therefore, we first study sparse MRC portfolios.
Our goal is to construct a sparse vector of portfolio weights given by \eqref{e158}. To achieve this we use the following equivalent and unconstrained regression representation of the mean-variance optimization in \eqref{sys2} and \eqref{sys3}:
\begin{equation}\label{ee65}
\bw_{MRC}=\argmin_{\bw}\E{y-\bw'\br_t}, \quad \text{where} \quad y\defeq \dfrac{1+\theta}{\theta}\mu \equiv \sigma\dfrac{1+\theta}{\sqrt{\theta}}.
\end{equation}
The sample counterpart of \eqref{ee65} is written as:
\begin{equation}
\bw_{MRC}=\argmin_{\bw}\frac{1}{T}\sum_{t=1}^{T}(y-\bw'\br_{t})^{2}.
\end{equation}
\cite{Ao2019} prove that the weight allocation from \eqref{ee65} is equivalent to \eqref{e158}. The sparsity is introduced through Lasso which yields the following constrained optimization problem:
\begin{equation}\label{ee67}
\bw_{\text{MRC, SPARSE}}=\argmin_{\bw}\frac{1}{T}\sum_{t=1}^{T}(y-\bw'\br_{t})^{2}+2\lambda\norm{\bw}_1.
\end{equation}

Now we propose two extensions to the setup \eqref{ee67}. First, the estimator $\bw_{\text{MRC, SPARSE}}$ is infeasible since $\theta$ used for constructing $y$ is unknown. \cite{Ao2019} construct an estimator of $\theta$ under normally distributed excess returns, assuming $p/T \rightarrow \rho \in (0,1)$ and the sample size $T$ is required to be larger than the number of assets $p$. Their paper uses an unbiased estimator proposed in \cite{Kan2007optimal}: $\hat{\theta}=((T-p-2)\widehat{\boldm}'\widehat{\bSigma}^{-1}\widehat{\boldm}-p)/T$, where $\widehat{\boldm}$ and $\widehat{\bSigma}^{-1}$ are sample mean and inverse of the sample covariance matrix respectively. One of the limitations of the model studied by \cite{Ao2019} is that it cannot handle high dimensions. In both simulations and empirical application the maximum number of stocks used by the authors is limited to 100. Another limitation of \cite{Ao2019} approach is that they do not correct the bias introduced by imposing $\ell_1$-constraint in \eqref{ee67}. However, it is well-known that the estimator in \eqref{ee67} is biased and the existing literature proposes several de-biasing techniques (see \cite{Zhang2014,Javanmard2014confidence,Javanmard2014hypothesis,Buhlmann2014,Belloni2015uniform,Javanmard2018debiasing} among others).

To address the first aforementioned limitation, we propose to use an estimator of a high-dimensional precision matrix discussed in the next section. The suggested estimator is appropriate for high-dimensional settings, it can handle cases when the sample size is less than the number of assets, and it is always non-negative by construction\footnote{Our empirical results suggest that the unbiased estimator $\hat{\theta}=((T-p-2)\widehat{\boldm}'\widehat{\bSigma}^{-1}\widehat{\boldm}-p)/T$ is oftentimes negative even after using the adjusted estimator defined in \cite{Kan2007optimal} (p. 2906).}. Consequently, the estimator of $y$ is
\begin{equation}\label{e2.17}
\widehat{y}\defeq \dfrac{1+\hat{\theta}}{\hat{\theta}}\mu \equiv \sigma\dfrac{1+\hat{\theta}}{\sqrt{\hat{\theta}}}.
\end{equation}

To approach the second limitation, motivated by \cite{Buhlmann2014}, we propose the de-biasing technique that uses the nodewise regression estimator of the precision matrix. First, let $\bR$ be a $T \times p$ matrix of excess returns stacked over time and $	\widehat{\by}$ be a $T \times 1$ constant vector. Consider a high-dimensional linear model
\begin{equation} \label{e4.9}
\widehat{\by}=\bR\bw+\be,\quad \text{where} \quad \be \sim \mathcal{D}(\bm{0},\sigma_{e}^{2}\bI).
\end{equation}
 We study high-dimensional framework $p\geq T$ and in the asymptotic results we require $\log p/T=o(1)$. Let us rewrite \eqref{ee67}:
\begin{equation}\label{ee69}
\bw_{\text{MRC, SPARSE}}=\argmin_{\bw \in \mathbb{R}^p}\frac{1}{T}\norm{\widehat{\by}-\bR\bw}_{2}^{2}+2\lambda\norm{\bw}_1.
\end{equation}
The estimator in \eqref{ee67} satisfies the following KKT conditions:
\begin{equation}\label{ee70}
-\bR'(\widehat{\by}-\bR\widehat{\bw})/T + \lambda\widehat{\bg}=0,
\end{equation}
\begin{equation}
\norm{\widehat{\bg}}_{\infty}\leq 1 \quad \text{and} \quad \hat{g_i}=\text{sign}(\hat{w}_{i}) \quad \text{if} \quad \hat{w}_{i}\neq 0.
\end{equation}
where $\widehat{\bg}$ is a $p \times 1$ vector arising from the subgradient of $\norm{\bw}_{1}$.
Let $\widehat{\bSigma}=\bR'\bR/T$, then we can rewrite the KKT conditions:
\begin{equation}\label{ee72}
\widehat{\bSigma}(\widehat{\bw}-\bw)+\lambda \widehat{\bg}=\bR'\be/T.
\end{equation}
Multiply both sides of \eqref{ee72} by $\widehat{\bTheta}$ obtained from Algorithm \ref{alg3}, add and subtract $(\widehat{\bw}-\bw)$, and rearrange the terms:
\begin{equation}\label{ee73}
\widehat{\bw}-\bw+\widehat{\bTheta}\lambda\widehat{\bg} = \widehat{\bTheta}\bR'\be/T-\underbrace{\sqrt{T}(\widehat{\bTheta}\widehat{\bSigma}-\bI_p)(\widehat{\bw}-\bw)}_{\Delta}/\sqrt{T}.
\end{equation}
In the section with the theoretical results we show that $\Delta$ is asymptotically negligible under certain sparsity assumptions\footnote{Note that we cannot directly apply Theorem 2.2 of \cite{Buhlmann2014} since $\br_c$ needs to be estimated and we first need to show consistency of the respective estimator.}. Combining \eqref{ee70} and \eqref{ee73} brings us to the de-biased estimator of portfolio weights:
\begin{equation}\label{ee74}
\widehat{\bw}_{\text{MRC, DEBIASED}}=\widehat{\bw}+\widehat{\bTheta}\lambda\widehat{\bg} = \widehat{\bw}+\widehat{\bTheta}\bR'(\widehat{\by}-\bR\widehat{\bw})/T.
\end{equation}
The properties of the proposed de-biased estimator are examined in Section 5.

\subsection{Sparse Portfolio Using Post-Lasso}
One of the drawbacks of the de-biased portfolio weights in \eqref{ee74} is that the weight formula is tailored to a specific portfolio choice that maximizes an unconstrained Sharpe Ratio (i.e. MRC in \eqref{e158}). However, it is desirable to accommodate preferences of different types of investors who might be interested in weight allocations corresponding to GMV \eqref{e156} or MWC \eqref{e157} portfolios. At the same time, we are willing to stay within the framework of sparse allocations. One of the difficulties that precludes us from pursuing a similar technique as in \eqref{ee67} is the fact that once the weight constraint is added, the optimization problem in \eqref{ee67} has two solutions depending on whether $\biota'\bTheta\boldm$ is positive or negative. As shown in \cite{maller2003new}, when $\biota'\bTheta\boldm <0$, the minimum value cannot be achieved exactly for a specified portfolio allocation that satisfies the full investment constraint. Hence, one can design an approximate solution to approach the supremum as closely as desired.

 To overcome this difficulty, we propose to use Lasso regression in \eqref{ee69} for selecting a subset of stocks, and then constructing a financial portfolio using any of the weight formulations in \eqref{e156}-\eqref{e158}. The procedure to estimate sparse portfolio using post-Lasso is described in Algorithm \ref{algPL}.
	\begin{algorithm}[H]
	\caption{Sparse Portfolio Using Post-Lasso}
	\label{algPL}
	\begin{algorithmic}[1]
		\STATE Use Lasso regression in \eqref{ee69} to select the model $\widehat{\Xi}\defeq \text{support}(\widehat{\bw})$
		\begin{itemize}
			\item Apply additional thresholding to remove stocks with small estimated weights:
			\begin{equation*}
			\widehat{\bw}(t)=(\widehat{w}_j\1{\abs{\widehat{w}_j }>t}, \ j=1,\ldots,p),
			\end{equation*}
			where $t\geq 0$ is the thresholding level.
			\item The corresponding selected model is denoted as $\widehat{\Xi}(t)\defeq \text{support}(\widehat{\bw}(t))$. When $t=0$, $\widehat{\Xi}(t)=\widehat{\Xi}$.
		\end{itemize}
		\STATE Choose a desired portfolio formulation in \eqref{e156}-\eqref{e158} and apply it to the selected subset of stocks $\widehat{\Xi}(t)$.
		\begin{itemize}
			\item When $\text{card}(\widehat{\Xi}(t)) < \widetilde{t}$, use the inverse of the sample covariance matrix as an estimator of $\bTheta$. Otherwise, apply the estimator of precision matrix described in Section 3.
		\end{itemize}
	\end{algorithmic}
\end{algorithm}	

	\section{Factor Nodewise Regression}
	In this section we first review a nodewise regression (\cite{meinshausen2006}), a popular approach to estimate a precision matrix. After that we propose a novel estimator which accounts for the common factors in the excess returns.
	
	In the high-dimensional settings it is necessary to regularize the precision matrix, which means that some of the entries $\theta_{ij}$  will be zero. In other words, to achieve consistent estimation of the inverse covariance, the estimated precision matrix should be sparse.
	\subsection{Nodewise Regression}
One of the approaches to induce sparsity in the estimation of precision matrix is to solve for $\widehat{\bTheta}$ one column at a time via linear regressions, replacing population moments by their sample counterparts. When we repeat this procedure for each variable $j=1,\ldots,p$, we will estimate the elements of $\widehat{\bTheta}$ column by column using $\{\br_t\}_{t=1}^{T}$ via $p$ linear regressions. \cite{meinshausen2006} use this approach to incorporate sparsity into the estimation of the precision matrix. They fit $p$ separate Lasso regressions using each variable (node) as the response and the others as predictors to estimate $\widehat{\bTheta}$. This method is known as the \enquote{nodewise} regression and it is reviewed below based on \cite{Buhlmann2014} and \cite{Caner2019}.
	
	Let $\br_j$ be a $T \times 1$ vector of observations for the $j$-th regressor, the remaining covariates are collected in a $T \times (p-1)$ matrix $\bR_{-j}$.  For each $j=1,\ldots,p$ we run the following Lasso regressions:
	\begin{equation}\label{e21}
	\widehat{\bgamma}_j= \argmin_{\bgamma\in \mathbb{R}^{p-1}}\Big(\norm{\br_j-\bR_{-j}\bgamma}_{2}^{2}/T+2\lambda_j\norm{\bgamma}_1\Big),
	\end{equation}
	where $\widehat{\bgamma}_j=\{\widehat{\gamma}_{j,k}; j=1,\ldots,p, k\neq j\}$ is a $(p-1)\times 1$ vector of the estimated regression coefficients that will be used to construct the estimate of the precision matrix, $\widehat{\bTheta}$. Define
	\begin{equation}
	\widehat{\bC}=\begin{pmatrix}
	1&-\widehat{\gamma}_{1,2}&\cdots&-\widehat{\gamma}_{1,p}\\
	-\widehat{\gamma}_{2,1}&1&\cdots&-\widehat{\gamma}_{2,p}\\
	\vdots&\vdots&\ddots&\vdots\\
	-\widehat{\gamma}_{p,1}&-\widehat{\gamma}_{p,2}&\cdots&1\\
	\end{pmatrix}.
	\end{equation} 
	For $j=1,\ldots,p$, define
	\begin{equation}
	\hat{\tau}_{j}^{2} = \norm{\br_j-\bR_{-j}\widehat{\bgamma}_j}_{2}^{2}/T+\lambda_j\norm{\widehat{\bgamma}_j}_1
	\end{equation}
	and write
	\begin{equation}
	\widehat{\bT}^2 = \text{diag}(\hat{\tau}_{1}^{2},\ldots,\hat{\tau}_{p}^{2}).
	\end{equation}
	The approximate inverse is defined as
	\begin{equation} \label{e25}
	\widehat{\bTheta} = \widehat{\bT}^{-2} \widehat{\bC}.
	\end{equation}
	The procedure to estimate the precision matrix using nodewise regression is summarized in Algorithm \ref{alg1b}.
	\begin{algorithm}[H]
		\caption{Nodewise Regression by \cite{meinshausen2006} (MB)}
		\label{alg1b}
		\begin{algorithmic}[1]
			\STATE Repeat for $j=1,\ldots,p$ :
			\begin{itemize}
				\item Estimate $\widehat{\bgamma}_j$ using \eqref{e21} for a given $\lambda_j$.
				\item  Select $\lambda_j$ using a suitable information criterion.
			\end{itemize}
			\STATE Calculate $\widehat{\bC}$ and $\widehat{\bT}^2$ .
			\STATE Return $\widehat{\bTheta} = \widehat{\bT}^{-2}  \widehat{\bC}$.
		\end{algorithmic}
	\end{algorithm}	
	One of the caveats to keep in mind when using the nodewise regression method is that the estimator in \eqref{e25} is not self-adjoint. \cite{Caner2019} show (see their Lemma A.1) that $\widehat{\bTheta}$ in \eqref{e25} is positive definite with high probability, however, it could still occur that $\widehat{\bTheta}$ is not positive definite in finite samples. To resolve this issue we use the matrix symmetrization procedure as in \cite{fan2018elliptical} and then use eigenvalue cleaning as in \cite{Callot2017} and \cite{Hautsch2012}. First, the symmetric matrix is constructed as
	\begin{equation}\label{symmetr}
	\widehat{\theta}^{s}_{ij}=\widehat{\theta}_{ij} \1{\abs{\widehat{\theta}_{ij}}\leq \abs{\widehat{\theta}_{ji}}}+\widehat{\theta}_{ji} \1{\abs{\widehat{\theta}_{ij}} > \abs{\widehat{\theta}_{ji}}},
	\end{equation}
	where $\widehat{\theta}_{ij}$ is the $(i,j)$-th element of the estimated precision matrix from \eqref{e25}. Second, we use eigenvalue cleaning to make $\widehat{\bTheta}^{s}$ positive definite: write the spectral decomposition $\widehat{\bTheta}^{s}=\widehat{\bV}'\widehat{\bLambda}\widehat{\bV}$, where $\widehat{\bV}$ is a matrix of eigenvectors and $\widehat{\bLambda}$ is a diagonal matrix with $p$ eigenvalues $\widehat{\bLambda}_{i}$ on its diagonal. Let $\bLambda_{m}\defeq\min\{\widehat{\bLambda}_{i}|\widehat{\bLambda}_{i}>0\}$. We replace all $\widehat{\bLambda}_{i}<\bLambda_{m}$ with $\bLambda_{m}$ and define the diagonal matrix with cleaned eigenvalues as $\widetilde{\bLambda}$. We use $\widetilde{\bTheta}=\widehat{\bV}'\widetilde{\bLambda}\widehat{\bV}$ which is symmetric and positive definite.
	
	\subsection{Factor Nodewise Regression}
	The arbitrage pricing theory (APT), developed by \cite{APTRoss}, postulates that expected returns on securities should be related to their covariance with the common components or factors only. The goal of the APT is to model the tendency of asset returns to move together via factor decomposition. Assume that the return generating process ($\br_t$) follows a $K$-factor model:
	\begin{align} \label{equ1}
	&\underbrace{\br_t}_{p \times 1}=\bB \underbrace{\bf_t}_{K\times 1}+ \ \bvarepsilon_t,\quad t=1,\ldots,T
	\end{align}
	where $\bf_t=(f_{1t},\ldots, f_{Kt})'$ are the factors, $\bB$ is a $p \times K$ matrix of factor loadings, and $\bvarepsilon_t$ is the idiosyncratic component that cannot be explained by the common factors. Factors in \eqref{equ1} can be either observable, such as in \cite{Fama3Factor,Fama5Factor}, or can be estimated using statistical factor models.
	
	In this subsection we examine how to approach the portfolio allocation problems in \eqref{e156}-\eqref{e158} using a factor structure in the returns. Our approach, called \textit{Factor Nodewise Regression}, uses the estimated common factors to obtain sparse precision matrix of the idiosyncratic component. The resulting estimator is used to obtain the precision of the asset returns necessary to form portfolio weights.
	
	As in \cite{fan2013POET}, we consider a spiked covariance model when the first $K$ principal eigenvalues of $\bSigma$ are growing with $p$, while the remaining $p-K$ eigenvalues are bounded and grow slower than $p$.
	
	Rewrite equation \eqref{equ1} in matrix form:
	\begin{equation} \label{5.2n}
	\underbrace{\bR}_{p\times T} = \underbrace{\bB}_{p\times K} \bF + \bE.
	\end{equation}
	
	\noindent Let $\bSigma=T^{-1}\bR\bR'$, $\bSigma_{\varepsilon}=T^{-1}\bE\bE'$ and $\bSigma_{f}=T^{-1}\bF\bF'$ be covariance matrices of stock returns, idiosyncratic components and factors, and let $\bTheta=\bSigma^{-1}$, $\bTheta_{\varepsilon}=\bSigma_{\varepsilon}^{-1}$ and $\bTheta_{f}=\bSigma_{f}^{-1}$ be their inverses. The factors and loadings in \eqref{5.2n} are estimated by solving $	(\widehat{\bB},\widehat{\bF})=\argmin_{\bB,\bF}\norm{\bR-\bB\bF}^{2}_{F}$ s.t. $\frac{1}{T}\bF\bF'=\bI_K, \ \bB'\bB\ \text{is diagonal}$. The constraints are needed to identify the factors (\cite{fan2018elliptical}). It was shown (\cite{Stock2002}) that $\widehat{\bF}=\sqrt{T}\text{eig}_K(\bR'\bR)$ and $\widehat{\bB}=T^{-1}\bR\widehat{\bF}'$. Given $\widehat{\bF},\widehat{\bB}$, define $\widehat{\bE}=\bR-\widehat{\bB}\widehat{\bF}$.
	
	Since our interest is in constructing portfolio weights, our goal is to estimate a precision matrix of the excess returns. However, as pointed out by \cite{koike2019biased}, when common factors are present across the excess returns, the precision matrix cannot be sparse because all pairs of the returns are partially correlated given other excess returns through the common factors. Therefore, we impose a sparsity assumption on the precision matrix of the idiosyncratic errors, $\bTheta_{\varepsilon}$, which is obtained using the estimated residuals after removing the co-movements induced by the factors (see \cite{Brownlees2018EJS,Brownlees2018JAE,koike2019biased}).
	
	We use the nodewise regression as a shrinkage technique to estimate the precision matrix of residuals. Once the precision $\bTheta_{f}$ of the low-rank component is also obtained, similarly to \cite{Fan2011}, we use the Sherman-Morrison-Woodbury formula to estimate the precision of excess returns:
	\begin{equation}\label{equa18}
	\bTheta=\bTheta_{\varepsilon}-\bTheta_{\varepsilon}\bB\lbrack\bTheta_f+\bB'\bTheta_{\varepsilon}\bB\rbrack^{-1}\bB'\bTheta_{\varepsilon}.
	\end{equation}
	To obtain $\widehat{\bTheta}_{f}=\widehat{\bSigma}_{f}^{-1}$, we use the inverse of the sample covariance of the estimated factors $\widehat{\bSigma}_{f}=T^{-1}\widehat{\bF}\widehat{\bF}'$. To get $\widehat{\bTheta}_{\varepsilon}$, we apply Algorithm \ref{alg1b} to the estimated idiosyncratic errors, $\widehat{\bvarepsilon}_t$.  Once we have estimated $\widehat{\bTheta}_{f}$ and $\widehat{\bTheta}_{\varepsilon}$, we can get $\widehat{\bTheta}$ using a sample analogue of \eqref{equa18}. The proposed procedure is called \textit{Factor Nodewise Regression} and is summarized in Algorithm \ref{alg3}.

	\begin{algorithm}[H]
		\caption{Factor Nodewise Regression by \cite{meinshausen2006} (FMB)}
		\label{alg3}
		\begin{algorithmic}[1]
			\STATE Estimate factors, $\widehat{\bF}$, and factor loadings, $\widehat{\bB}$, using PCA. Obtain $\widehat{\bSigma}_{f}=T^{-1}\widehat{\bF}\widehat{\bF}'$, $\widehat{\bTheta}_{f}=\widehat{\bSigma}_{f}^{-1}$ and $\widehat{\bvarepsilon}_{t}=\br_t-\widehat{\bB}\widehat{\bf_t}$.
			\STATE Estimate a sparse $\bTheta_{\varepsilon}$ using nodewise regression: run Lasso regressions in \eqref{e21} for $\widehat{\bvarepsilon}_t$			
			\begin{equation}\label{e21a}
			\widehat{\bgamma}_j= \argmin_{\bgamma\in \mathbb{R}^{p-1}}\Big(\norm{\widehat{\bvarepsilon}_j-\widehat{\bE}_{-j}\bgamma}_{2}^{2}/T+2\lambda_j\norm{\bgamma}_1\Big),
			\end{equation}
			to get $\widehat{\bTheta}_{\varepsilon}$.
			\STATE Use $\widehat{\bTheta}_f$ from Step 1 and $\widehat{\bTheta}_{\varepsilon}$ from Step 2 to estimate $\bTheta$ using the sample counterpart of the Sherman-Morrison-Woodbury formula in \eqref{equa18}:
	\begin{equation}\label{3.11}
	\widehat{\bTheta}=\widehat{\bTheta}_{\varepsilon}-\widehat{\bTheta}_{\varepsilon}\widehat{\bB}\lbrack\widehat{\bTheta}_f+\widehat{\bB}'\widehat{\bTheta}_{\varepsilon}\widehat{\bB}\rbrack^{-1}\widehat{\bB}'\widehat{\bTheta}_{\varepsilon}.
	\end{equation}
		\end{algorithmic}
	\end{algorithm}	
	Algorithm \ref{alg3} involves a tuning parameter $\lambda_j$ in \eqref{e21a}: we choose shrinkage intensity by minimizing the generalized information criterion (GIC). Let $\abs{\widehat{S}_j(\lambda_j)}$ denote the estimated number of nonzero parameters in the vector $\widehat{\bgamma}_j$:
	\begin{equation*}
	\text{GIC}(\lambda_j) = \log\Big(\norm{\widehat{\bvarepsilon}_j-\widehat{\bE}_{-j}\widehat{\bgamma}_j}_{2}^{2}/T\Big)+\abs{\widehat{S}_j(\lambda_j)}\frac{\log(p)}{T}\log(\log(T)).
	\end{equation*}

	We can use $\widehat{\bTheta}$ obtained in \eqref{3.11} to estimate $y$ in equation \eqref{e2.17} and obtain sparse portfolio weights in \eqref{ee74} and Algorithm \ref{algPL}.

	\section{Factor Investing is Allowed}
	In this section we allow an investor to hold a portfolio of assets and factors, in other words, factors are assumed to be tradable. Note that in contrast with \cite{Ao2019}, the distinction between tradable and non-tradable factors is not pinned down by the fact that the excess returns are driven by the common factors. That is, factor structure of returns is allowed independently of whether factors are tradable or not. We assume that only observable factors can be tradable. Denote a $K_1 \times 1$ vector of observable factors as $\widetilde{\bf}_t$, and $K_2 \times 1$ vector of unobservable factors as $\bf^{PCA}_t$, where $K1+K2=K$.
	The goal of factor investing is to decide how much weight is allocated to factors $\widetilde{\bf}_t$ and stocks $\br_t$.
	Let $r_{t,all}$ be the return of portfolio at time $t$:
	\begin{equation} \label{e5.2}
	r_{t,all} = \underbrace{\bw'_{all,t}}_{1\times (p+K_1)} \bx_t.
	\end{equation}
	where $\bx_t=(\widetilde{\bf}'_t,\br'_t)'$ is a $(p+K_1)\times 1$ vector of excess returns of observable factors and stocks and $\bw_{all,t}=(\bw'_{ft}, \bw'_t)'$ is a vector of weights with $\bw_{ft}$ invested in $\widetilde{\bf}_t$ and $\bw_t$ invested in stocks. We treat $\widetilde{\bf}_t$ as additional $K_1$ investments vehicles which will contribute to the return of the total portfolio. Now consider $K_2$-factor model for $\bx_t$:
	\begin{align} \label{e6.2}
	&\bx_t=\bB \underbrace{\bf^{PCA}_t}_{K_2\times 1}+ \ \be_t,\quad t=1,\ldots,T
	\end{align}
	
	Rewrite equation \eqref{e6.2} in matrix form:
	\begin{equation} \label{e6.3}
	\underbrace{\bX}_{(p+K_1)\times T} = \bB \underbrace{\bF^{PCA}}_{K_2\times T} + \ \bE,
	\end{equation}
	which can be estimated using the standard PCA techniques as in \eqref{5.2n}:\\ $\widehat{\bF}^{PCA}=\sqrt{T}\text{eig}_{K_2}(\bX'\bX)$ and $\widehat{\bB}=T^{-1}\bX\widehat{\bF}^{'PCA}$. Given $\widehat{\bF}^{PCA},\widehat{\bB}$, define $\widehat{\bE}=\bX-\widehat{\bB}\widehat{\bF}^{PCA}$.
	
	Similarly to Algorithm \ref{alg3}, we use \eqref{equa18} to estimate the precision of the augmented excess returns, $\bTheta_x$. To get $\widehat{\bTheta}_{f^{PCA}}=\widehat{\bSigma}_{f^{PCA}}^{-1}$, we use the inverse of the sample covariance of the estimated factors $\widehat{\bSigma}_{f^{PCA}}=T^{-1}\widehat{\bF}^{PCA}\widehat{\bF}^{'PCA}$. To get $\widehat{\bTheta}_{e}$, we first apply Algorithm \ref{alg1b} to the estimated idiosyncratic errors, $\widehat{\be}_t$ in \eqref{e6.2}.  Once we have estimated $\widehat{\bTheta}_{f^{PCA}}$ and $\widehat{\bTheta}_{e}$, we can get $\widehat{\bTheta}_x$ using a sample analogue of \eqref{equa18}. This procedure is summarized in Algorithm \ref{alg4}.
		\begin{algorithm}[H]
		\caption{Factor Investing Using FMB}
		\label{alg4}
		\begin{algorithmic}[1]
		\STATE Estimate the residuals from equation \eqref{e6.2}: $\widehat{\be}_{t}=\bx_t-\widehat{\bB}\widehat{\bf}_{t}^{PCA}$ using PCA.
		\STATE Estimate a sparse $\bTheta_{e}$ using nodewise regression: apply Algorithm \ref{alg1b} to $\widehat{e}_t$.
		\STATE Estimate $\bTheta_x$ using the Sherman-Morrison-Woodbury formula in \eqref{equa18}.
	\end{algorithmic}
\end{algorithm}	
	We can use $\widehat{\bTheta}_x$ obtained from Algorithm \ref{alg4} to estimate portfolio weights $\bw_{all,t}$ using either a de-biased technique from section 2.1 (\eqref{ee74}), or post-Lasso (Algorithm \ref{algPL}). Once we obtain $\widehat{\bw}_{all,t}=(\widehat{\bw}'_{ft}, \widehat{\bw}'_t)'$, we can test whether factor investing significantly contributes to the portfolio return by testing whether $\bw_{ft}=0$.

	\section{Asymptotic Properties}
In this section we study asymptotic properties of the de-biased estimator of weights for sparse portfolio in \eqref{ee74} and post-Lasso estimator from Algorithm \ref{algPL}.

Denote $S_0\defeq \{ j; \bw_j\neq 0 \}$ to be the active set of variables, where $\bw$ is a vector of true portfolio weights in equation \eqref{e4.9}. Also, let $s_0\defeq \abs{S_0}$. Further, let  $S_j\defeq \{ k; \gamma_{j,k}\neq 0 \}$ be the active set for row $\bgamma_{j}$ for the nodewise regression in \eqref{e21}, and let $s_j\defeq \abs{S_j}$. Define $\bar{s}\defeq \max_{1\leq j\leq p}s_j$.

Consider a factor model from equation \eqref{equ1}:
\begin{align} \label{e5.1}
&\underbrace{\br_t}_{p \times 1}=\bB \underbrace{\bf_t}_{K\times 1}+ \ \bvarepsilon_t,\quad t=1,\ldots,T
\end{align}
 We study the case when the factors are not known, i.e. the only observable variable in equation \eqref{e5.1} is the excess returns $\br_t$. In this paper our main interest lies in establishing asymptotic properties of sparse portfolio weights and the out-of-sample Sharpe Ratio for the high-dimensional case. We assume that the number of common factors, $K$, is fixed.

\subsection{Assumptions}
 We now list the assumptions on the model \eqref{e5.1}:
\begin{enumerate}[\textbf{({A}.1)}]
	\item \label{A1} (Spiked covariance model)
As $p \rightarrow \infty$, $\Lambda_1(\bSigma)>\Lambda_2(\bSigma)>\ldots>\Lambda_K(\bSigma)\gg \Lambda_{K+1}(\bSigma)\geq \ldots \geq \Lambda_p(\bSigma) \geq 0$, where $\Lambda_j(\bSigma)=\mathcal{O}(p)$ for $j \leq K$, while the non-spiked eigenvalues are bounded, $\Lambda_j(\bSigma)=o(p)$ for $j > K$.
\end{enumerate}
\begin{enumerate}[\textbf{({A}.2)}]
	\item \label{A2}(Pervasive factors)
	There exists a positive definite $K \times K$ matrix $\breve{\bB}$ such that\\ $\vertiii{p^{-1}\bB'\bB-\breve{\bB}}_{2}\rightarrow 0$ and $\Lambda_{\text{min}}(\breve{\bB})^{-1}=\mathcal{O}(1)$  as $p \rightarrow \infty$.
\end{enumerate}
Similarly to \cite{CHANG2018} and \cite{Caner2019}, we also impose beta mixing condition.
\begin{enumerate}[\textbf{({A}.3)}]
	\item \label{A3} (Beta mixing)
 Let $\mathcal{F}_{-\infty}^{t}$ and $\mathcal{F}_{t+k}^{\infty}$ denote the $\sigma$-algebras that are generated by $\{\bvarepsilon_u:u\leq t\}$ and $\{\bvarepsilon_u:u\geq t+k\}$ respectively. Then $\{\bvarepsilon\}_u$ is $\beta$-mixing in the sense that $\beta_k\rightarrow 0$ as $k \rightarrow \infty$, where the mixing coefficient is defined as
\begin{equation}
\beta_k=\sup_t\E{\sup_{B\in \mathcal{F}_{t+k}^{\infty}} \abs{\Pr\Big(B| \mathcal{F}_{-\infty}^{t}\Big)-\Pr\Big(B\Big)} }.
\end{equation}
\end{enumerate}

Some comments regarding the aforementioned assumptions are in order. Assumptions \ref{A1}-\ref{A2} are the same as in \cite{fan2018elliptical}, and assumption \ref{A3} is required to consistently estimate precision matrix for de-biasing portfolio weights. Assumption \ref{A1} divides the eigenvalues into the diverging and bounded ones. Without loss of generality, we assume that $K$ largest eigenvalues have multiplicity of 1. The assumption of a spiked covariance model is common in the literature on approximate factor models, however, we note that the model studied in this paper can be characterized as a \enquote{very spiked model}. In other words, the gap between the first $K$ eigenvalues and the rest is increasing with $p$. As pointed out by \cite{fan2018elliptical}, \ref{A1} is typically satisfied by the factor model with pervasive factors, which brings us to the assumption \ref{A2}: the factors impact a non-vanishing proportion of individual time-series. Assumption \ref{A3} allows for weak dependence in the residuals of the factor model in \ref{e5.1}: causal ARMA processes, certain stationary Markov chains and stationary GARCH models with finite second moments satisfy this assumption. We note that our Assumption \ref{A3} is much weaker than in \cite{Caner2019}, the latter requires weak dependence of the returns series, whereas we only restrict dependence of the idiosyncratic components.

Let $\bSigma=\bGamma\bLambda\bGamma^{'}$, where $\bSigma$ is the covariance matrix of returns that follow factor structure described in equation \eqref{e5.1}. Define $\widehat{\bSigma}, \widehat{\bLambda}_K,\widehat{\bGamma}_K$ to be the estimators of $\bSigma,\bLambda,\bGamma$. We further let $\widehat{\bLambda}_K=\text{diag}(\hat{\lambda}_1,\ldots,\hat{\lambda}_K)$ and $\widehat{\bGamma}_K=(\hat{v}_1,\ldots,\hat{v}_K)$ to be constructed by the first $K$ leading empirical eigenvalues and the corresponding eigenvectors of $\widehat{\bSigma}$ and $\widehat{\bB}\widehat{\bB}'=\widehat{\bGamma}_K\widehat{\bLambda}_K\widehat{\bGamma}_{K}^{'}$. Similarly to \cite{fan2018elliptical}, we require the following bounds on the componentwise maximums of the estimators:
\begin{enumerate}[\textbf{({B}.1)}]
	\item \label{B1} $\norm{\widehat{\bSigma}-\bSigma}_{\text{max}}=\mathcal{O}_P\Big(\sqrt{\log p/T}\Big)$,
\end{enumerate}

\begin{enumerate}[\textbf{({B}.2)}]
	\item \label{B2} $\norm{(\widehat{\bLambda}_K-\bLambda)\bLambda^{-1}}_{\text{max}}=\mathcal{O}_P\Big(\sqrt{\log p/T}\Big)$,
\end{enumerate}
\begin{enumerate}[\textbf{({B}.3)}]
	\item \label{B3} $\norm{\widehat{\bGamma}_K-\bGamma}_{\text{max}}=\mathcal{O}_P\Big(\sqrt{\log p/(Tp)})$.
\end{enumerate}

Let $\widehat{\bSigma}^{SG}$ be the sample covariance matrix, with $\widehat{\bLambda}_{K}^{SG}$ and $\widehat{\bGamma}_{K}^{SG}$ constructed with the first $K$ leading empirical eigenvalues and eigenvectors of $\widehat{\bSigma}^{SG}$ respectively. Also, let $\widehat{\bSigma}^{EL1} = \widehat{\bD}\widehat{\bR}_1\widehat{\bD}$, where $\widehat{\bR}_1$ is obtained using the Kendall's tau correlation coefficients and $\widehat{\bD}$ is a robust estimator of variances constructed using the Huber loss. Furthermore, let $\widehat{\bSigma}^{EL2} = \widehat{\bD}\widehat{\bR}_2\widehat{\bD}$, where $\widehat{\bR}_2$ is obtained using the spatial Kendall's tau estimator. Define $\widehat{\bLambda}_{K}^{EL}$ to be the matrix of the first $K$ leading empirical eigenvalues of $\widehat{\bSigma}^{EL1}$, and $\widehat{\bGamma}_{K}^{EL}$ is the matrix of the first $K$ leading empirical eigenvectors of $\widehat{\bSigma}^{EL2}$. For more details regarding constructing $\widehat{\bSigma}^{SG}$, $\widehat{\bSigma}^{EL1}$ and $\widehat{\bSigma}^{EL2}$ see \cite{fan2018elliptical}, Sections 3 and 4.
\begin{thm}\label{theor1Fan} (\cite{fan2018elliptical})\\
	For sub-Gaussian distributions, $\widehat{\bSigma}^{SG}$, $\widehat{\bLambda}_{K}^{SG}$ and $\widehat{\bGamma}_{K}^{SG}$ satisfy \ref{B1}-\ref{B3}.\\
	For elliptical distributions, $\widehat{\bSigma}^{EL1}$, $\widehat{\bLambda}_{K}^{EL}$ and $\widehat{\bGamma}_{K}^{EL}$ satisfy \ref{B1}-\ref{B3}.
\end{thm}
Theorem \ref{theor1Fan} is essentially a rephrasing of the results obtained in \cite{fan2018elliptical}, Sections 3 and 4. Since there is no separate statement of these results in their paper (it is rather a summary of several theorems), we separated it as a Theorem for the convenience of the reader. As evidenced from the above Theorem, $\widehat{\bSigma}^{EL2}$ is only used for estimating the eigenvectors. This is necessary due to the fact that, in contrast with $\widehat{\bSigma}^{EL2}$, the theoretical properties of the eigenvectors of $\widehat{\bSigma}^{EL}$ are mathematically involved because of the sin function. 

In addition, the following structural assumption on the model is imposed:
\begin{enumerate}[\textbf{({C}.1)}]
	\item \label{C1} $\norm{\bSigma}_{\text{max}}=\mathcal{O}(1)$ and $\norm{\bB}_{\text{max}}=\mathcal{O}(1)$,
\end{enumerate}

which is a natural assumption on the population quantities.

In contrast to \cite{fan2018elliptical}, instead of estimating and inverting covariance matrix, we focus on obtaining precision matrix directly since it is the ultimate input to any portfolio optimization problem.

\subsection{Asymptotic Properties of Non-Sparse Portfolio Weights}
Recall that we used equation \eqref{equa18} to estimate $\bTheta$.
Therefore, in order to establish consistency of the estimator in \eqref{equa18}, we first show consistency of $\widehat{\bTheta}_{\varepsilon}$. Proofs of all the theorems are in Appendix.
\begin{thm}\label{theor2}
	Suppose that Assumptions \ref{A1}-\ref{A3}, \ref{B1}-\ref{B3} and \ref{C1} hold. Let $\omega_{T}\defeq \sqrt{\log p/T} +1/\sqrt{p}$. Then $\max_{i\leq p} (1/T)\sum_{t=1}^{T} \abs{\hat{\varepsilon}_{it}-\varepsilon_{it}}=\mathcal{O}_P(\omega_{T}^2)$ and $\max_{i,t}\abs{\hat{\varepsilon}_{it}-\varepsilon_{it}}=\mathcal{O}_P(\omega_{T})=o_P(1)$. Under the sparsity assumption $\bar{s}^2\omega_{T}=o(1)$, with $\lambda_j \asymp \omega_{T}$, we have
	\begin{align*}
	&\max_{1\leq j \leq p} \norm{\widehat{\bTheta}_{\varepsilon,j}-\bTheta_{\varepsilon,j}}_1=\mathcal{O}_P(\bar{s}\omega_{T}),\\
	&\max_{1\leq j \leq p} \norm{\widehat{\bTheta}_{\varepsilon,j}-\bTheta_{\varepsilon,j}}_{2}^{2}=\mathcal{O}_P(\bar{s}\omega_{T}^2)
	\end{align*}
\end{thm}
Some comments are in order. First, the sparsity assumption $\bar{s}^2\omega_{T}=o(1)$ is stronger than that required for convergence of $\widehat{\bTheta}_{\varepsilon}$: this is necessary to ensure consistency for $\widehat{\bTheta}$ established in Theorem \ref{theor3}, so we impose a stronger assumption at the beginning. We also note that at the first glance, our sparsity assumption in Theorem \ref{theor3} is stronger than that required by \cite{Buhlmann2014} and \cite{Caner2019}, however, recall that we impose sparsity on $\bTheta_{\varepsilon}$, not $\bTheta$ as opposed to the two aforementioned papers. Hence, this assumption can be easily satisfied once the common factors have been accounted for and the precision of the idiosyncratic components is expected to be sparse. The bounds derived in Theorem \ref{theor2} help us establish the convergence properties of the precision matrix of stock returns in equation \eqref{equa18}.
\begin{thm} \label{theor3}
Under the assumptions of Theorem \ref{theor2} and, in addition, assuming $\norm{\bTheta_{\varepsilon,j}}_2=\mathcal{O}(1)$, we have 
\begin{align*}
&\max_{1\leq j \leq p} \norm{\widehat{\bTheta}_{j}-\bTheta_{j}}_1=\mathcal{O}_P(\bar{s}^2\omega_{T}),\\
&\max_{1\leq j \leq p} \norm{\widehat{\bTheta}_{j}-\bTheta_{j}}_{2}^{2}=\mathcal{O}_P(\bar{s}\omega_{T}^2).
\end{align*}
\end{thm}
Using Theorem \ref{theor3} we can then establish the consistency of the non-sparse counterpart of the estimated MRC portfolio weight in \eqref{ee69}.
\begin{thm} \label{theor3a}
Under the assumptions of Theorem \ref{theor3}, Algorithm \ref{alg3} consistently estimates non-sparse MRC portfolio weights such that $\norm{\widehat{\bw}_{\text{MRC}}-\bw_{\text{MRC}}}_1 = \mathcal{O}_P( \bar{s}^{2}\omega_{T}^{1/2})$.
\end{thm}
Note that the rate in Theorem \ref{theor3a} depends on the sparsity of $\bTheta_{\varepsilon}$. If, instead, sparsity on $\bTheta$ is imposed, the rate becomes similar to the one derived by \cite{Caner2019}: $\bar{s}(\bTheta)^{3/2}\omega_{T}^{1/2}=o_P(1)$, where $\bar{s}(\bTheta)$ is the maximum vertex degree of $\bTheta$. In their case, if the precision matrix of stock returns is not sparse, consistent estimation of portfolio weights is possible if $(p-1)^{3/2}(\sqrt{\log p/T}+1/\sqrt{p})=o(1)$. However, this excludes high-dimensional cases since $p$ is required to be less than $T^{1/3}$.

\subsection{Asymptotic Properties of De-Biased Portfolio Weights}
We now proceed to examining the properties of sparse MRC portfolio weights for de-biased portfolio, as summarized by the following Theorem:
\begin{thm} \label{theor4}
	Let $\widehat{\bSigma}$ be an estimator of covariance matrix satisfying \ref{B1}, and $\widehat{\bTheta}$ be the estimator of precision obtained using FMB in Algorithm \ref{alg3}.
	Under the assumptions of Theorem \ref{theor3}, consider the linear model \eqref{e4.9} with $\be \sim \mathcal{D}(\bm{0},\sigma_{e}\bI)$, where $\sigma_{e}^{2}=\mathcal{O}(1)$. Consider a suitable choice of the regularization parameters $\lambda \asymp \omega_{T}$ for the Lasso regression in \eqref{ee69} and $\lambda_j \asymp \omega_{T}$ uniformly in $j$ for the Lasso for nodewise regression in \eqref{e21}. Assume $(s_0\vee \bar{s}^2)\Big(\log p/\sqrt{T}+\sqrt{T}/p\Big)=o(1)$. Then
	\begin{align*}
	\sqrt{T}(\widehat{\bw}_{\text{DEBIASED}}-&\bw) = W+\Delta,\\
	&W = \widehat{\bTheta}\bR'\be/\sqrt{T},\\
	&\norm{\Delta}_{\infty} = \mathcal{O}_P\Big( (s_0\vee \bar{s}^2)\Big(\log p/\sqrt{T}+\sqrt{T}/p\Big)\Big)=o_P(1).
	\end{align*}
	Furthermore, if $\be \sim \mathcal{N}(\bm{0},\sigma_{e}^{2}\bI)$, let $\widehat{\bOmega}\defeq \widehat{\bTheta}\widehat{\bSigma}\widehat{\bTheta}'$. Then $W|\bR \sim \mathcal{N}_p(\mathbf{0},\sigma_{e}^{2}\widehat{\bOmega})$ and $\norm{\widehat{\bOmega}-\bTheta}_{\infty}=o_P(1)$.
\end{thm}
Some comments are in order. Our Theorem \ref{theor4} is an extension of Theorem 2.4 of \cite{Buhlmann2014} for non-iid case, where the latter is achieved with a help of \cite{CHANG2018}.
 Furthermore, there are several fundamental differences between Theorem \ref{theor4} and Theorem 2.4 of \cite{Buhlmann2014}: first, we apply nodewise regression to estimate sparse precision matrix of factor-adjusted returns, which explains the difference in convergence rates. Concretely, \cite{Buhlmann2014}  have $\omega_{T}=\sqrt{\log p/T}$, whereas we have $\omega_{T}=\sqrt{\log p/T}+1/\sqrt{p}$, where $1/\sqrt{p}$ arises due to the fact that factors need to be estimated. However, we note that since we deal with high-dimensional regime $p\geq T$, this additional term is asymptotically negligible, we only keep it for identification purposes. Second, in contrast with \cite{Buhlmann2014}, the dependent variable in the Lasso regression in \eqref{ee69} is unknown and needs to be estimated. Lemma \ref{lemA1caner} shows that $\widehat{y}$ constructed using the precision matrix estimator from Theorem \ref{theor3} is consistent and shares the same rate as the $\ell_1$-bound in Theorem \ref{theor3}. Third, interestingly, the sparsity assumption on the Lasso regression in \eqref{ee69} is the same as in \cite{Buhlmann2014}: as shown in the Appendix, this condition is still sufficient to ensure that the bias term is asymptotically negligible even when the stock returns follow factor structure with unknown factors. Once we impose Gaussianity of $\be$ in \eqref{e4.9}, we can infer the distribution of portfolio weights. Note that in this case normally distributed errors do not imply that the stock returns are also Gaussian: we did not assume $\varepsilon_t \sim \mathcal{N}_p(\cdot)$ in \eqref{e5.1}. The unknown $\sigma_{e}^{2}$ can be replaced by a consistent estimator.
 Finally, even when Gaussianity of $\be$ is relaxed, we can use the central limit theorem argument to obtain approximate Gaussianity of components of $W|\bR$ of fixed dimension, or moderately growing dimensions (see \cite{Buhlmann2014} for more details), however, in order not to divert the focus of this paper, we leave it for future research.
 \subsection{Asymptotic Properties of Post-Lasso Portfolio Weights}
 To establish the properties of the post-Lasso estimator in Algorithm \ref{algPL}, we focus on MRC weight formulation, since it satisfies the standard post-Lasso assumptions. For GMV and MWC formulations, the procedure described in Algorithm \ref{algPL} is not \enquote{post-Lasso} in the usual sense. Concretely, the latter assumes that both steps in Algorithm \ref{algPL} have the same objective function, which is violated for GMV and MWC. Consequently, we leave rigorous theoretical derivations of these two portfolio formulation for future research. For MRC, we use the post-model selection results established in \cite{chernozhukov2013}. Specifically, we have the following theorem:
 \begin{thm}\label{theor5}
 	Suppose the restricted eigenvalue condition and the restricted sparse eigenvalue condition on the empirical Gram matrix hold (see Condition RE($\bar{c}$) and Condition RSE($m$) of \cite{chernozhukov2013}, p. 529). Let $\widehat{\bw}$ be the post-Lasso weight estimator from Algorithm \ref{algPL}, we have
 	\begin{align*}
 	\norm{\widehat{\bw}-\bw}_1 = \mathcal{O}_P \begin{cases}
 	\sigma_{e} \Big((s_0\omega_{T})\vee (\bar{s}^2\omega_{T})\Big), & \text{in general,}\\
 	\sigma_{e} s_0 \Big(\sqrt{\frac{1}{T}}+\frac{1}{\sqrt{p}}\Big), & \text{if} \ s_0\geq \bar{s}^2 \ \text{and} \ \Xi=\hat{\Xi} \ \textup{wp} \rightarrow 1.\\
 	\end{cases}
 	\end{align*}
 \end{thm}
The proof of Theorem \ref{theor5} easily follows from the proof of Corollary 2 of \cite{chernozhukov2013} and is omitted here. Let us comment on the upper bounds for post-Lasso estimator: first, the term $(s_0\omega_{T})\vee (\bar{s}^2\omega_{T})$ appears since one needs to estimate the dependent variable in equation \eqref{e4.9}, which creates the difference between the bound in \cite{chernozhukov2013} and our Theorem \ref{theor5}. Second, similarly to \cite{chernozhukov2013}, the upper bound undergoes a transition from the oracle rate enjoyed by the standard Lasso to the faster rate that improves on the latter when (1) the precision matrix of the idiosyncratic components is sparse enough and (2) the oracle model has well-separated coefficients. Noticeably, the upper bounds in Theorem \ref{theor5} hold despite the fact that the first-stage Lasso regression in Algorithm \ref{algPL} may fail to correctly select the oracle model $\Xi$ as a subset, that is, $\Xi \notin \hat{\Xi}$. 

Finally, let us compare the rates of non-sparse MRC portfolio weights in Theorem \ref{theor3a}, de-biased weights in Theorem \ref{theor4}, and post-lasso weights in Theorem \ref{theor5}: de-biased estimator exhibits fastest convergence, followed by post-lasso and non-sparse weights. This result is further supported by our simulations presented in the next section.
\section{Monte Carlo}
 We study the consistency for estimating portfolio weights in \eqref{ee20} of (i) sparse portfolios that use the standard Lasso without de-biasing in \eqref{ee69}, (ii) Lasso with de-biasing in \eqref{ee74}, (iii) post-Lasso in Algorithm \ref{algPL}, and (iv) non-sparse portfolios that use FMB from Algorithm \ref{alg3}. Our simulation results are divided into two parts: the first part examines the performance of models (i)-(iv) under the Gaussian setting, and the second part examines the robustness of performance under the elliptical distributions (to be described later). Each part is further subdivided into two cases: with $p<T$ (\textbf{Case 1}) and with $p>T$ (\textbf{Case 2}), in both cases we allow the number of stocks to increase with the sample size, i.e. $p=p_T \rightarrow \infty$ as $T \rightarrow \infty$. In Case 1 we let $p = T^{\delta}$, $\delta = 0.85$ and $T = \lbrack 2^h \rbrack, \ \text{for} \ h=7,7.5,8,\ldots,9.5$, in Case 2 we let $p = 3\cdot T^{\delta}$, $\delta = 0.85$, all else equal.
 
 First, consider the following data generating process for stock returns:
 \begin{align} 
 \underbrace{\br_t}_{p \times 1}=\boldm + \bB \underbrace{\bf_t}_{K\times 1}+ \ \bvarepsilon_t, \quad t=1,\ldots,T \label{e52}
 \end{align}
 where $\boldm_i \sim \mathcal{N}(1,1)$ independently for each $i=1,\ldots,p$, $\bvarepsilon_{t}$ is a $p \times 1$ random vector of idiosyncratic errors following $\mathcal{N}(\bm{0},\bSigma_{\varepsilon})$, with a Toeplitz matrix $\bSigma_{\varepsilon}$ parameterized by $\rho$: that is, $\bSigma_\varepsilon = (\bSigma_\varepsilon)_{ij}$, where $(\bSigma_\varepsilon)_{ij}=\rho^{\abs{i-j}}$, $i,j\in 1,\ldots,p$ which leads to sparse $\bTheta_{\varepsilon}$, $\bf_{t}$ is a $K \times 1$ vector of factors drawn from $\mathcal{N}(\mathbf{0},\bSigma_f = \bI_{K}/10)$, $\bB$ is a $p\times K$ matrix of factor loadings drawn from $\mathcal{N}(\mathbf{0},\bI_{K}/100)$. We set $\rho = 0.5$ and fix the number of factors $K=3$.
 
Let $\bSigma = \bB\bSigma_f\bB'+\bSigma_{\varepsilon}$. To create sparse MRC portfolio weights we use the following procedure: first, we threshold the vector $\bSigma^{-1}\boldm$ to keep the top $p/2$ entries with largest absolute values. This yields sparse vector $\balpha=\bSigma^{-1}\boldm$ defined in \eqref{e158}. We use $\bSigma\balpha$ and $\bSigma$ as the values for the mean and covariance matrix parameters to generate multivariate Gaussian returns in \eqref{e52}. Note that the low rank plus sparse structure of the covariance matrix is preserved under this transformation.

Figure \ref{f1} shows the averaged (over Monte Carlo simulations) errors of the estimators of the weight $\bw_{\text{MRC}}$ versus the sample size $T$ in the logarithmic scale (base 2). As evidenced by Figure \ref{f1}, (1) sparse estimators outperform non-sparse counterparts; (2) using de-biasing or post-Lasso improves the performance compared to the standard Lasso estimator. As expected from Theorems \ref{theor4}-\ref{theor5}, the Lasso, de-biased Lasso and post-Lasso exhibit similar rates, but the two latter estimators enjoy lower estimation error. The ranking remains similar for Case 2, however, as illustrated in Figure \ref{f1}, the performance of all estimators slightly deteriorates.

Gaussian-tail assumption is too restrictive for modeling the behavior of financial returns. Hence, as a second exercise we check the robustness of our sparse portfolio allocation estimators under the elliptical distributions, which we briefly review based on \cite{fan2018elliptical}. Elliptical distribution family generalizes the multivariate normal distribution and multivariate t-distribution. Let $\boldm \in \mathbb{R}^p$ and $\bSigma \in \mathbb{R}^{p\times p}$. A $p$-dimensional random vector $\br$ has an elliptical distribution, denoted by $\br \sim \text{ED}_p(\boldm,\bSigma,\zeta)$, if it has a stochastic representation
\begin{equation}\label{eq6.2}
\br \stackrel{d}{=} \boldm + \zeta \bA \bU,
\end{equation}
where $\bU$ is a random vector uniformly distributed on the unit sphere $\mathcal{S}^{q-1}$ in $\mathbb{R}^q$, $\zeta \geq 0$ is a scalar random variable independent of $\bU$, $\bA \in \mathbb{R}^{p\times q}$ is a deterministic matrix satisfying $\bA\bA'=\bSigma$. As pointed out in \cite{fan2018elliptical}, the representation in \eqref{eq6.2} is not identifiable, hence, we require $\E{\zeta^2}=q$, such that $\text{Cov}(\br) = \bSigma$. We only consider continuous elliptical distributions with $\Pr[\zeta=0]=0$. The advantage of the elliptical distribution for the financial returns is its ability to model heavy-tailed data and the tail dependence between variables.

Having reviewed the elliptical distribution, we proceed to the second part of simulation results. The data generating process is similar to \cite{fan2018elliptical}: let $(\bf_t, \bvarepsilon_{t})$ from \eqref{e52} jointly follow the multivariate t-distribution with the degrees of freedom $\nu$. When $\nu = \infty$, this corresponds to the multivariate normal distribution, smaller values of $\nu$ are associated with thicker tails. We draw $T$ independent samples of $(\bf_t, \bvarepsilon_{t})$ from the multivariate t-distribution with zero mean and covariance matrix $\bSigma = \text{diag}(\bSigma_{f}, \bSigma_{\varepsilon})$, where $\bSigma_{f} = \bI_{K}$. To construct $\bSigma_{\varepsilon}$ we use a Toeplitz structure parameterized by $\rho = 0.5$, which leads to the sparse $\bTheta_{\varepsilon} = \bSigma^{-1}_{\varepsilon}$. The rows of $\bB$ are drawn from $\mathcal{N}(\bm{0}, \bI_{K}/100)$. Figure 2 reports the results for $\nu=4.2$\footnote{The results for larger degrees of freedom do not provide any additional insight, hence we do not report them here. However, they are available upon request.}: the performance of the standard Lasso estimator significantly deteriorates, which is further amplified in the high-dimensional case where it exhibits the worst performance. Noticeably, post-Lasso still achieves the lowest estimation error, followed by de-biased estimator.

\section{Empirical Application}
This section is divided into three main parts. First, we examine the performance of several non-sparse portfolios, including the equal-weighted and Index portfolios (reported as the composite S\&P500 index listed as \textsuperscript{$\wedge$}GSPC). Second, we study the performance of sparse portfolios that are based on de-biasing and post-Lasso. Third, we consider several interesting periods that include different states of the economy: we examine the merit of sparse vs non-sparse portfolios during the periods of economic growth, moderate market decline and severe economic downturns.
\subsection{Data}
We use monthly returns of the components of the S\&P500 index\footnote{The conclusions from using daily data are the same as those for monthly returns, hence we do not report them in the main manuscript text. However, they are available upon request.}. The data on historical S\&P500 constituents and stock returns is fetched from CRSP and Compustat using SAS interface. The full sample has 480 observations on 355 stocks from January 1, 1980 - December 1, 2019. We use January 1, 1980 - December 1, 1994 (180 obs) as a training period and January 1, 1995 - December 1, 2019 (300 obs) as the out-of-sample test period.
We roll the estimation window over the test sample to rebalance the portfolios monthly. At the end of each month, prior to portfolio construction, we remove stocks with less than 15 years of historical stock return data. For sparse portfolio we employ the following strategy to choose the tuning parameter $\lambda$ in \eqref{ee67}: we use the first two thirds of the training data (which we call the training window) to estimate weights and tune the shrinkage intensity $\lambda$ in the remaining one third of the training sample to yield the highest Sharpe Ratio which serves as a validation window. We estimate factors and factor loadings in the training window and validation window combined. The risk-free rate and Fama-French factors are taken from \href{https://mba.tuck.dartmouth.edu/pages/faculty/ken.french/data_library.html}{Kenneth R. French's data library.}

\subsection{Performance Measures}
Similarly to \cite{Caner2019}, we consider four metrics commonly reported in finance literature: the Sharpe Ratio, the portfolio turnover, the average return and risk of a portfolio. We consider two scenarios: with and without transaction costs. Let $T$ denote the total number of observations, the training sample consists of $m$ observations, and the test sample is $n=T-m$. When transaction costs are not taken into account, the out-of-sample average portfolio return, risk and Sharpe Ratio are
\begin{align}
& \hat{\mu}_{\text{test}}=\frac{1}{n}\sum_{t=m}^{T-1}\widehat{\bw}'_{t}\br_{t+1},\\
& \hat{\sigma}_{\text{test}}=\sqrt{\frac{1}{n-1}\sum_{t=m}^{T-1}(\widehat{\bw}'_{t}\br_{t+1}-\hat{\mu}_{\text{test}})^2},\\
& \text{SR} = \hat{\mu}_{\text{test}}/ \hat{\sigma}_{\text{test}}.
\end{align}
We follow \cite{Demiguel2009optimal,Caner2019,Li2015sparse,ban2018machine} to account for transaction costs (tc). In line with the aforementioned papers, we set $c=50 \text{bps}$. Define the excess portfolio at time $t+1$ with transaction costs as 
\begin{align}
r_{t+1,\text{portfolio}} = &\widehat{\bw}'_{t}\br_{t+1}-c(1+\widehat{\bw}'_{t}\br_{t+1})\sum_{j=1}^{p}\abs{\hat{w}_{t+1,j}-\hat{w}_{t,j}^{+}},\\
\text{where}\quad&\hat{w}_{t,j}^{+}=\hat{w}_{t,j}\frac{1+r_{t+1,j}+r^{f}_{t+1}}{1+r_{t+1,\text{portfolio}}+r^{f}_{t+1}},
\end{align} 
where $r_{t+1,j}+r^{f}_{t+1}$ is sum of the excess return of the $j$-th asset and risk-free rate, and $r_{t+1,\text{portfolio}}+r^{f}_{t+1}$ is the sum of the excess return of the portfolio and risk-free rate. The out-of-sample average portfolio return, risk, Sharpe Ratio and turnover are defined accordingly:
\begin{align}
& \hat{\mu}_{\text{test,tc}}=\frac{1}{n}\sum_{t=m}^{T-1}r_{t,\text{portfolio}},\\
& \hat{\sigma}_{\text{test,tc}}=\sqrt{\frac{1}{n-1}\sum_{t=m}^{T-1}(r_{t,\text{portfolio}}-\hat{\mu}_{\text{test,tc}})^2},\\
& \text{SR}_{\text{tc}} =\hat{\mu}_{\text{test,tc}}/ \hat{\sigma}_{\text{test,tc}},\\
&\text{Turnover}=\frac{1}{n}\sum_{t=m}^{T-1}\sum_{j=1}^{p}\abs{\hat{w}_{t+1,j}-\hat{w}_{t,j}^{+}}.
\end{align}
\subsection{Results and Discussion}
The first set of results explores the performance of several non-sparse portfolios: equal-weighted portfolio (EW), Index portfolio (Index), MB from Algorithm \ref{alg1b}, FMB from Algorithm \ref{alg3}, linear shrinkage estimator of covariance that incorporates factor structure through the Sherman-Morrison inversion formula (\cite{Ledoit2004}, further referred to as LW), CLIME (\cite{cai2011constrained}). We consider two scenarios, when the factors are unknown and estimated using the standard PCA (statistical factors), and when the factors are known. For the statistical factors, we  determine the number of factors, $K$, in a standard data-driven way using the information criteria discussed in \cite{Bai2002} and \cite{kapetanios2010testing} among others. For the scenario with known factors we include up to 5 Fama-French factors: FF1 includes the excess return on the market, FF3 includes FF1 plus size factor (Small Minus Big, SMB) and value factor (High Minus Low, HML), and FF5 includes FF3 plus profitability factor (Robust Minus Weak, RMW) and risk factor (Conservative Minus Agressive, CMA). In Table \ref{tab1}, we report monthly portfolio performance for three alternative portfolio allocations in \eqref{e156}, \eqref{e157} and \eqref{e156}. We set a return target $\mu= 0.7974\%$ which is equivalent to $10\%$ yearly return when compounded. The target level of risk for the weight-constrained and risk-constrained Markowitz portfolio (MWC and MRC) is set at $\sigma=0.05$ which is the standard deviation of the monthly excess returns of the S\&P500 index in the first training set.

We now comment on the results which are presented in Table \ref{tab1}: \textbf{(1)} accounting for the factor structure in stock returns improves portfolio performance in terms of the OOS Sharpe Ratio. Specifically, EW, Index, MB and CLIME which ignore factor structure perform worse than FMB and LW. \textbf{(2)} The models that use an improved estimator of covariance or precision matrix outperform EW and Index on the test sample. As a downside, such models have higher Turnover. This implies that superior performance is achieved at the cost of larger variability of portfolio positions over time and, as a consequence, increased risk associated with it.

The second set of results studies the performance of sparse portfolios: we include our proposed methods based on de-biasing and post-Lasso, as well as the approach studied in \cite{Ao2019} (Lasso) without factor investing. For post-Lasso we first use Lasso-based weight estimator in \eqref{ee69} for selecting stocks with absolute value of weights above a small threshold $\epsilon$ (we use $\epsilon=0.0001$), then we form portfolio with the selected stocks using three alternative portfolio allocations in \eqref{e156}-\eqref{e158}. 

Let us comment on the results presented in Table \ref{tab3}: \textbf{(1)}  column one demonstrates that de-biasing leads to significant performance improvement in terms of the return and the OOS Sharpe Ratio. Note that even though the risk of de-biased portfolio is also higher, it still satisfies the risk-constraint. This result emphasizes the importance of correcting for the bias introduced by the $\ell_1$-regularization. \textbf{(2)} Comparing two bias-correction methods, de-biasing and post-Lasso, we find that the latter is characterized by higher return and higher risk. However, increase in portfolio return brought by post-Lasso is, overall, not sufficient to outperform de-biasing approach in terms of the out-of-sample Sharpe Ratio. \textbf{(3)} Sparse portfolios have lower return, risk and turnover compared to non-sparse counterparts in Table \ref{tab1}, however, the OOS Sharpe Ratio is comparable, i.e. we do not see uniform superiority of either method. Therefore, incorporating sparsity allows investors to reduce portfolio risk at the cost of lower return while maintaining the Sharpe Ratio comparable to holding a non-sparse portfolio. 

Tables \ref{tabtry}-\ref{tabtrysparse} compare the performance of non-sparse and sparse (de-biased. ``DL", and post-Lasso, ``PL") portfolios for different time periods in terms of the cumulative excess return (CER) over the period of interest and risk. The first period of interest (1997-98), which will be referred to as ``Period I", corresponds to economic growth since Index exhibited positive CER during this time. The second period of interest, ``Period II", corresponds to moderate market decline since EW and Index had relatively small negative CER. Finally, ``Period III", corresponds to severe economic downturn and significant drop in the performance of EW and Index. We note that the references to the specific crises in Tables \ref{tabtry}-\ref{tabtrysparse} do not intend to limit these economic periods to these time spans. They merely provide the context for the time intervals of interest. Since the performance of MWC portfolios is similar to GMV, we only report MRC and GMV for the ease of presentation.

Let us summarize the findings from Tables \ref{tabtry}-\ref{tabtrysparse}: \textbf{(1)} In Period I non-sparse portfolios that rely on the estimation of covariance or precision matrix outperformed EW and Index in terms of CER for both MRC and GMV. However, in Period II GMV portfolios exhibited slightly negative CER, whereas MRC portfolios had higher risk but positive CER (albeit being lower compared to Period I). Note that in Period III none of the non-sparse portfolios generated positive  CER and portfolio risk increased rapidly. Examining the performance of sparse portfolios in Table \ref{tabtrysparse}, we see that \textbf{(2)} our proposed sparse portfolios produce positive CER during all three periods of interest. Furthermore, the return generated by PL is higher than that by non-sparse portfolios even during Periods I and II. Interestingly, DL produces positive CER without having high risk exposure. This suggests that our de-biased estimator of portfolio weights exhibits minimax properties. We leave the formal theoretical treatment of the latter for the future research. 

\section{Conclusion}
This paper develops an approach to construct sparse portfolios in high dimensions that addresses the shortcomings of the existing sparse portfolio allocation techniques.  We establish the oracle bounds of sparse weight estimators and provide guidance regarding their distribution. From the empirical perspective, we examine the merit of sparse portfolios during different market scenarios. We find that in contrast to non-sparse counterparts, our strategy is robust to recessions and can be used as a hedging vehicle during such times. Our framework makes use of the tool from the network theory called nodewise regression which not only satisfies desirable statistical properties, but also allows us to study whether certain industries could serve as safe havens during recessions. We find that such non-cyclical industries as consumer staples, healthcare, retail and food were driving the returns of the sparse portfolios during both the global financial crisis of 2007-09 and COVID-19 outbreak, whereas insurance sector was the least attractive investment in both periods. Finally, we develop a simple framework that provides clear guidelines how to implement factor investing using the methodology developed in this paper.
\section{Acknowledgements}
I greatly appreciate thoughtful comments and immense support from Tae-Hwy Lee, Jean Helwege, Jang-Ting Guo, Aman Ullah, Matthew Lyle, Varlam Kutateladze and UC Riverside Finance faculty. I also thank seminar participants at the 14th International CFE Conference (virtual), 2021 Southwestern Finance Association Annual Meeting, and Vilnius University.
\cleardoublepage

\phantomsection

\addcontentsline{toc}{section}{References}
\setlength{\baselineskip}{14pt}
\bibliographystyle{apalike}
\bibliography{SereginaSparsePortfolios}
%\nocite{*}
\cleardoublepage
%%%%%%%%%%%%%%%%%%%%%%%%%%%%%%%%%%%%%%%%%%%%%%%%%%%%%%%%%%%%%%%%%%%%%%%%%%%%%%%%%%%%%%%
% APPENDIX %%%%%%%%%%%%%%%%%%%%%%%%%%%%%%%%%%%%%%%%%%%%%%%%%%%%%%%%%%%%%%%%%%%%%%%%%%%%
%%%%%%%%%%%%%%%%%%%%%%%%%%%%%%%%%%%%%%%%%%%%%%%%%%%%%%%%%%%%%%%%%%%%%%%%%%%%%%%%%%%%%%%
\begin{sidewaysfigure}
	\begin{subfigure}{0.55\hsize}\centering
			\phantomsection
		\addcontentsline{toc}{section}{Figures}
		\includegraphics[width=1.08\hsize]{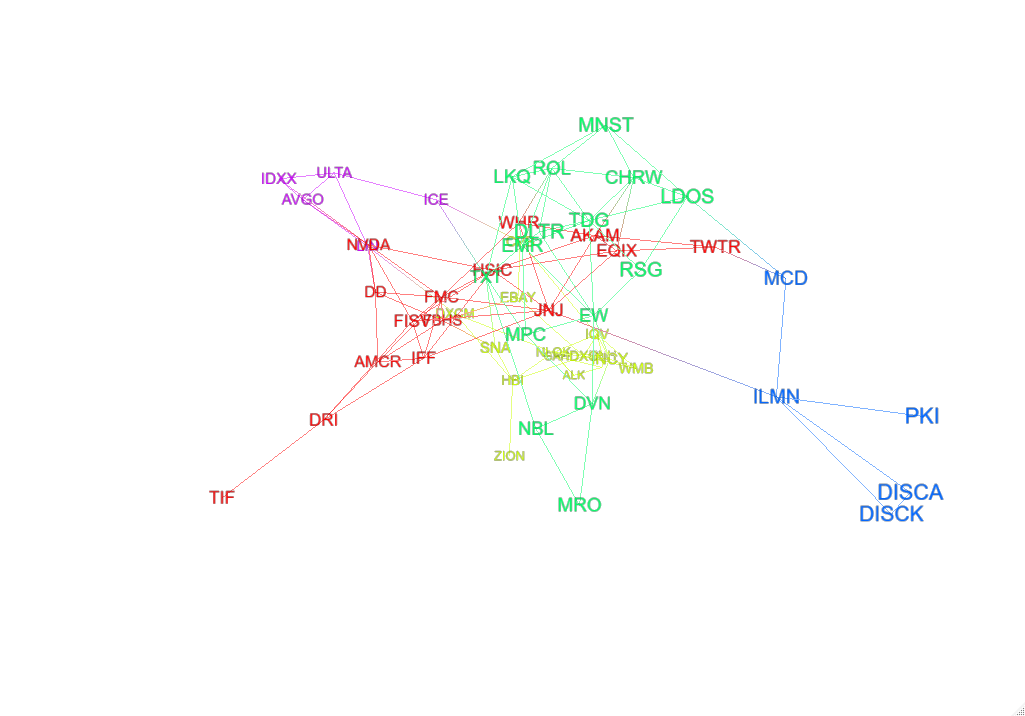}
		\caption{Stocks selected by Post-Lasso in August, 2019}
		\label{fig:sub1}
	\end{subfigure}%
%	\hfill %<-- it is superfluous 
\begin{subfigure}{0.55\hsize}\centering
		\includegraphics[width=1.08\hsize]{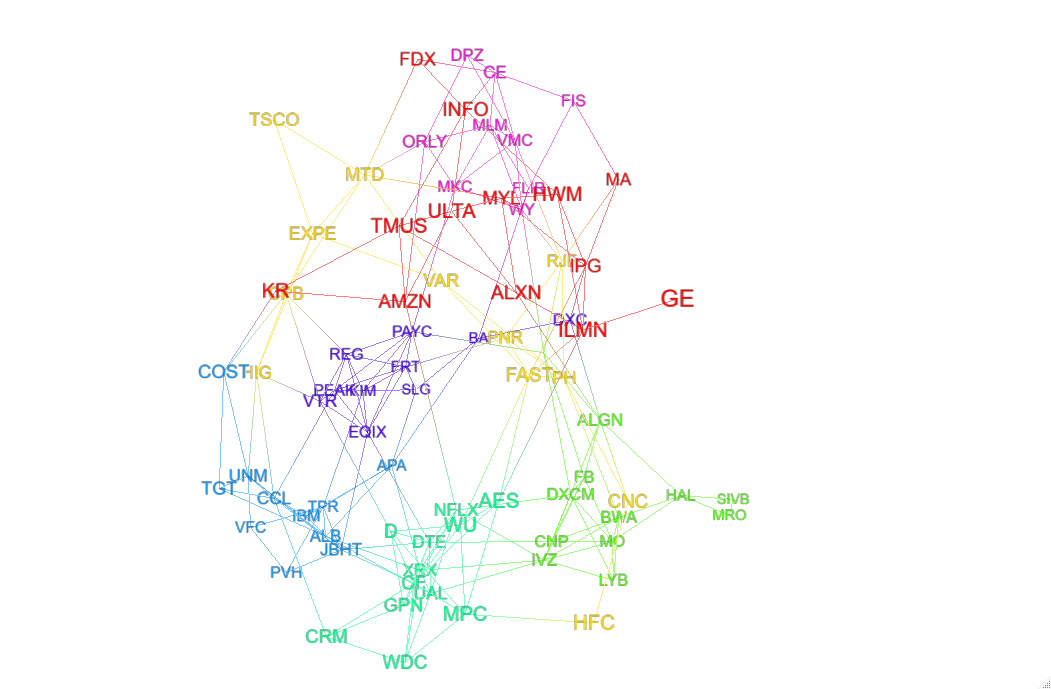}
		\caption{Stocks selected by Post-Lasso in May, 2020}
		\label{fig:sub2}
	\end{subfigure}
	\caption{\textit{Stocks selected by Post-Lasso strategy from Table \ref{table1}.}}
	\label{network}
\end{sidewaysfigure}

\begin{figure}[!htbp]
	\centering
	\includegraphics[width=0.98\textwidth]{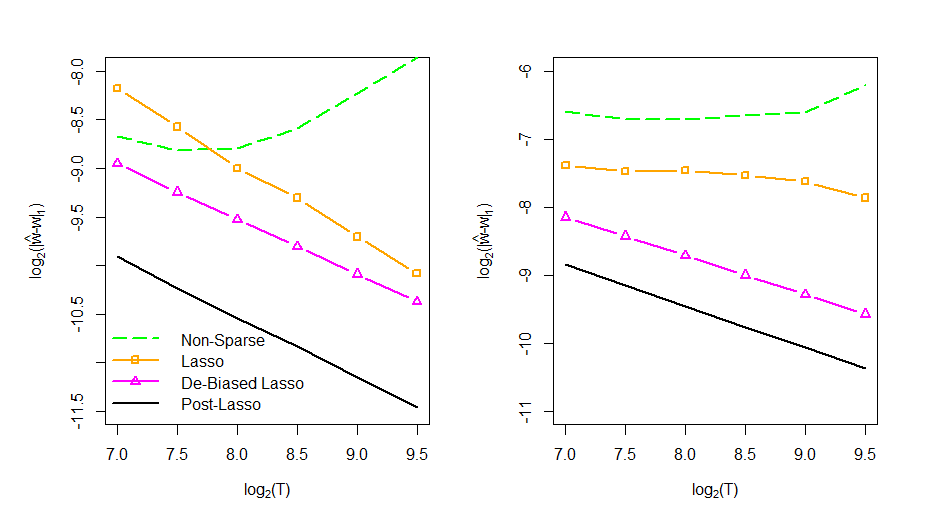}
	\bigskip
	\caption{\linespread{1.2} \selectfont\small \textit{Averaged errors of the estimators of $\bw_{\text{MRC}}$ for Case 1 on logarithmic scale (left): $p = T^{0.85}$, $K = 3$ and for Case 2 on logarithmic scale (right): $p =  3\cdot T^{0.85}$, $K = 3$.}}
	\label{f1}
\end{figure}
\begin{figure}[!htbp]
	\centering
	\includegraphics[width=0.98\textwidth]{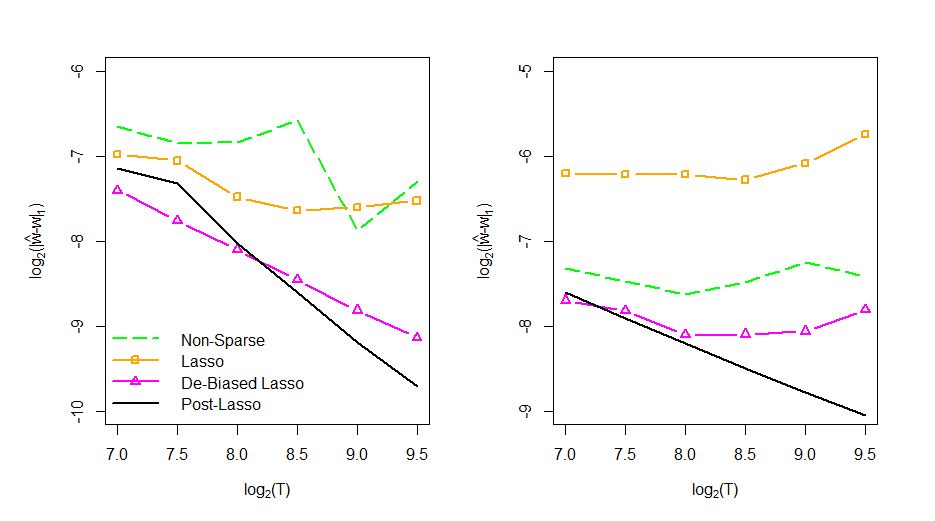}
	\bigskip
	\caption{\linespread{1.2} \selectfont\small \textit{Elliptical Distribution ($\nu=4.2$): Averaged errors of the estimators of $\bw_{\text{MRC}}$ for Case 1 on logarithmic scale (left): $p = T^{0.85}$, $K = 3$ and for Case 2 on logarithmic scale (right): $p =  3\cdot T^{0.85}$, $K = 3$.}}
	\label{f2}
\end{figure}
\begin{landscape}
	\begin{table}[]
		\caption{\textit{Monthly portfolio returns, risk, Sharpe Ratio and turnover. Transaction costs are set to 50 basis points, targeted risk is set at $\sigma=0.05$ (which is the standard deviation of the monthly excess returns on S\&P 500 index from 1980 to 1995, the first training period), monthly targeted return is $0.7974\%$ which is equivalent to $10\%$ yearly return when compounded. In-sample: January 1, 1980 - December 31, 1995 (180 obs), Out-of-sample: January 1, 1995 - December 31, 2019 (300 obs).}}
		\centering
		\label{tab1}
		\resizebox{0.8\textwidth}{!}{%
			\begin{tabular}{@{}ccccccccccccc@{}}
				\toprule
				& \multicolumn{4}{c}{Markowitz (risk-constrained)} & \multicolumn{4}{c}{Markowitz (weight-constrained)} & \multicolumn{4}{c}{Global Minimum-Variance} \\ \midrule
				& Return & Risk & SR & Turnover & Return & Risk & SR & Turnover & Return & Risk & SR & Turnover \\ \midrule
				\textbf{Without TC} &  &  &  & \multicolumn{1}{c}{} &  &  &  & \multicolumn{1}{c}{} &  &  &  &  \\
				EW & 0.0081 & 0.0520 & 0.1553 & \multicolumn{1}{c}{-} & 0.0081 & 0.0520 & 0.1553 & \multicolumn{1}{c}{-} & 0.0081 & 0.0520 & 0.1553 & - \\
				Index & 0.0063 & 0.0458 & 0.1389 & \multicolumn{1}{c}{-} & 0.0063 & 0.0458 &  0.1389 & \multicolumn{1}{c}{-} & 0.0063 & 0.0458 &  0.1389 & - \\
				MB & 0.0539 & 0.2522 & 0.2138 & \multicolumn{1}{c}{-} & 0.0070 & 0.0021 & 0.1539 & \multicolumn{1}{c}{-} & 0.0082 & 0.0020 & 0.1860 & - \\
				FMB (PC) & 0.0287 & 0.1049 & 0.2743 & \multicolumn{1}{c}{-} & 0.0069 & 0.0346 & 0.1968 & \multicolumn{1}{c}{-} & 0.0076 & 0.0346 & 0.2211 & - \\
				CLIME & 0.0372 & 0.2337 & 0.1593 & \multicolumn{1}{c}{-} & 0.0067 & 0.0471 & 0.1434
				& \multicolumn{1}{c}{-} & 0.0076 & 0.0466 & 0.1643 & - \\			
				LW & 0.0296 & 0.1049 & 0.2817 & \multicolumn{1}{c}{-} & 0.0059 & 0.0353 & 0.1662
				& \multicolumn{1}{c}{-} & 0.0063 & 0.0353 & 0.1774 & - \\
				FMB (FF1) & 0.0497 & 0.2200 & 0.2258 & \multicolumn{1}{c}{-} & 0.0071 & 0.0447 & 0.1582 & \multicolumn{1}{c}{-} & 0.0083 & 0.0436 & 0.1921 & - \\
				FMB (FF3) & 0.0384 & 0.1319 & 0.2908 & \multicolumn{1}{c}{-} & 0.0067 & 0.0387 & 0.1754 & \multicolumn{1}{c}{-} & 0.0080 & 0.0361 & 0.2223 & - \\
				FMB (FF5) & 0.0373 & 0.1277 & 0.2921 & \multicolumn{1}{c}{-} & 0.0068 & 0.0374 & 0.1788 & \multicolumn{1}{c}{-} & 0.0081 & 0.0361 & 0.2250 & - \\
				\midrule
				\textbf{With TC} &  &  &  & \multicolumn{1}{c}{} &  &  &  & \multicolumn{1}{c}{} &  &  &  &  \\
				EW & 0.0080 & 0.0520 & 0.1538 & \multicolumn{1}{c}{0.0630} & 0.0080 & 0.0027 & 0.1538 & \multicolumn{1}{c}{0.0630} & 0.0080 & 0.0027 & 0.1538 & 0.0630 \\
				MB & 0.0512 & 0.0637 & 0.2027 & \multicolumn{1}{c}{2.9458} & 0.0067 & 0.0021 & 0.1461 & \multicolumn{1}{c}{0.3223} & 0.0080 & 0.0020 & 0.1804 & 0.2152 \\
				FMB (PC) & 0.0248 & 0.1049 & 0.2368 & \multicolumn{1}{c}{3.7190} & 0.0059 & 0.0346 & 0.1687 & \multicolumn{1}{c}{0.9872} & 0.0067 & 0.0346 & 0.1929 & 0.9686 \\
				CLIME & 0.0334 & 0.2334 & 0.1429 & \multicolumn{1}{c}{4.9174} & 0.0062 & 0.0471 & 0.1307 & \multicolumn{1}{c}{0.5945} & 0.0071 & 0.0466 & 0.1522 &0.5528 \\	
				LW & 0.0237 & 0.1052 & 0.2257 & \multicolumn{1}{c}{5.5889} & 0.0043 & 0.0353 &
				0.1231 & \multicolumn{1}{c}{1.5166} &0.0048 & 0.0354 &  0.1343 & 1.5123\\
				FMB (FF1) & 0.0470 & 0.2202 & 0.2136 & \multicolumn{1}{c}{2.7245} & 0.0067 & 0.0447 & 0.1498 & \multicolumn{1}{c}{0.3489} & 0.0080 & 0.0436 & 0.1857 & 0.2486 \\
				FMB (FF3) & 0.0356 & 0.1319 & 0.2694 & \multicolumn{1}{c}{2.4670} & 0.0062 & 0.0387 & 0.1622 & \multicolumn{1}{c}{0.4728} & 0.0076 & 0.0361 & 0.2106 & 0.3920 \\
				FMB (FF5) & 0.0345 & 0.1277 & 0.2699 & \multicolumn{1}{c}{2.4853} & 0.0063 & 0.0387 & 0.1653 & \multicolumn{1}{c}{0.4847} & 0.0076 & 0.0361 & 0.2129 & 0.4057 \\
				\bottomrule
					\phantomsection
				\addcontentsline{toc}{section}{Tables}
			\end{tabular}%
		}
	\end{table}
	
\end{landscape}
\begin{landscape}
	\begin{table}[]
		\centering
		\caption{\textit{Sparse portfolio (FMB is used for de-biasing): monthly portfolio returns, risk, Sharpe Ratio and turnover. Transaction costs are set to 50 basis points, targeted risk is set at $\sigma=0.05$ (which is the standard deviation of the monthly excess returns on S\&P 500 index from 1980 to 1995, the first training period), monthly targeted return is $0.7974\%$ which is equivalent to $10\%$ yearly return when compounded. Factor Nodewise-regression estimator of precision matrix is used for de-biasing. Selection threshold for Post-Lasso is $\epsilon=0.0001$. In-sample: January 1, 1980 - December 31, 1995 (180 obs), Out-of-sample: January 1, 1995 - December 31, 2019 (300 obs).}}
		\label{tab3}
		\resizebox{\textwidth}{!}{%
			\begin{tabular}{@{}ccccccccccccccccc@{}}
				\toprule
				& \multicolumn{4}{c}{De-Biasing} & \multicolumn{12}{c}{Post-Lasso} \\ \midrule
				& \multicolumn{4}{c}{Markowitz (RC)} & \multicolumn{4}{c}{Markowitz (RC)} & \multicolumn{4}{c}{Markowitz (WC)} & \multicolumn{4}{c}{GMV} \\ \midrule
				& Return & Risk & SR & Turnover & Return & Risk & SR & Turnover & Return & Risk & SR & Turnover & Return & Risk & SR & Turnover \\ \midrule
	\textbf{Without TC} &  &  &  & \multicolumn{1}{c}{} &  & \multicolumn{1}{c}{} &  & \multicolumn{1}{c}{} &  & \multicolumn{1}{c}{} &  & \multicolumn{1}{c}{} &  & \multicolumn{1}{c}{} &  &  \\
	Lasso (PC0) & 0.0007 & 0.0048 & 0.1406 & \multicolumn{1}{c}{-} &  & \multicolumn{1}{c}{} &  & \multicolumn{1}{c}{} &  & \multicolumn{1}{c}{} &  & \multicolumn{1}{c}{} &  & \multicolumn{1}{c}{} &  &  \\
	De-biased Lasso (PC0) & 0.0023 & 0.0100 & 0.2266 & \multicolumn{1}{c}{-} & 0.0287 & 0.1217 & 0.2362 & \multicolumn{1}{c}{-} & -0.0174 & 0.4987 & -0.0350 & \multicolumn{1}{c}{-} & -0.0187 & 0.4941 & -0.0379 & - \\
	Lasso (PC) & 0.0006 & 0.0052 & 0.1122 & \multicolumn{1}{c}{-} &  &  &  & \multicolumn{1}{c}{} &  &  &  & \multicolumn{1}{c}{} &  &  &  &  \\
	De-biased Lasso (PC) & 0.0067 & 0.0265 & 0.2542 & \multicolumn{1}{c}{-} & 0.0290 & 0.1005 & 0.2882 & \multicolumn{1}{c}{-} & 0.0075 & 0.0624 & 0.1205 & \multicolumn{1}{c}{-} & 0.0087 & 0.0458 & 0.1901 & - \\
	Lasso (FF1) & 0.0007 & 0.0039 & 0.1902 & \multicolumn{1}{c}{-} &  &  &  & \multicolumn{1}{c}{} &  &  &  & \multicolumn{1}{c}{} &  &  &  &  \\
	De-biased Lasso (FF1) & 0.0109 & 0.0346 & 0.3213 & \multicolumn{1}{c}{-} & 0.0207 & 0.1192 & 0.1738 & \multicolumn{1}{c}{-} & 0.0031 & 0.2468 & 0.0124 & \multicolumn{1}{c}{-} & -0.0222 & 0.6047 & -0.0367 & - \\
	Lasso (FF3) & 0.0004 & 0.0040 & 0.1113 & \multicolumn{1}{c}{-} &  &  &  & \multicolumn{1}{c}{} &  &  &  & \multicolumn{1}{c}{} &  &  &  &  \\
	De-biased Lasso (FF3) & 0.0072 & 0.0265 & 0.2721 & \multicolumn{1}{c}{-} & 0.0157 & 0.1245 & 0.1263 & \multicolumn{1}{c}{-} & 0.0136 & 0.1153 & 0.1182 & \multicolumn{1}{c}{-} & -0.0226 & 0.6054 & -0.0373 & - \\
	Lasso (FF5) & 0.0002 & 0.0042 & 0.0577 & \multicolumn{1}{c}{-} &  &  &  & \multicolumn{1}{c}{} &  &  &  & \multicolumn{1}{c}{} &  &  &  &  \\
	De-biased Lasso (FF5) & 0.0073 & 0.0300 & 0.2467 & \multicolumn{1}{c}{-} & 0.0212 & 0.1127 & 0.1879 & \multicolumn{1}{c}{-} & 0.0093 & 0.0693 & 0.1342 & \multicolumn{1}{c}{-} & 0.0094 & 0.0980 & 0.0959 & - \\ \midrule
	\textbf{With TC} &  &  &  & \multicolumn{1}{c}{} &  & \multicolumn{1}{c}{} &  & \multicolumn{1}{c}{} &  & \multicolumn{1}{c}{} &  & \multicolumn{1}{c}{} &  & \multicolumn{1}{c}{} &  &  \\
	Lasso (PC0) & 0.0006 & 0.0049 & 0.1189 & \multicolumn{1}{c}{0.0719} &  & \multicolumn{1}{c}{} &  & \multicolumn{1}{c}{} &  & \multicolumn{1}{c}{} &  & \multicolumn{1}{c}{} &  & \multicolumn{1}{c}{} &  &  \\
	De-biased Lasso (PC0) & 0.0020 & 0.0100 & 0.1953 & \multicolumn{1}{c}{0.7952} & 0.0262 & 0.1212 & 0.2155 & \multicolumn{1}{c}{2.1249} & -0.0191 & 0.4990 & -0.0383 & \multicolumn{1}{c}{1.5373} & -0.0199 & 0.4945 & -0.0402 & 1.0737 \\
	Lasso (PC) & 0.0004 & 0.0052 & 0.0845 & \multicolumn{1}{c}{0.1136} &  &  &  & \multicolumn{1}{c}{} &  &  &  & \multicolumn{1}{c}{} &  &  &  &  \\
	De-biased Lasso (PC) & 0.0055 & 0.0265 & 0.2061 & \multicolumn{1}{c}{1.2113} & 0.0268 & 0.1005 & 0.2668 & \multicolumn{1}{c}{2.1756} & 0.0059 & 0.0624 & 0.0940 & \multicolumn{1}{c}{1.5777} & 0.0076 & 0.0458 & 0.1652 & 1.1026 \\
	Lasso (FF1) & 0.0006 & 0.0038 & 0.1654 & \multicolumn{1}{c}{0.0789} &  &  &  & \multicolumn{1}{c}{} &  &  &  & \multicolumn{1}{c}{} &  &  &  &  \\
	De-biased Lasso (FF1) & 0.0100 & 0.0346 & 0.2949 & \multicolumn{1}{c}{0.8298} & 0.0186 & 0.1192 & 0.1559 & \multicolumn{1}{c}{2.1589} & 0.0013 & 0.2458 & 0.0052 & \multicolumn{1}{c}{1.7374} & -0.0234 & 0.6056 & -0.0386 & 1.1113 \\
	Lasso (FF3) & 0.0003 & 0.0040 & 0.0852 & \multicolumn{1}{c}{0.0785} &  &  &  & \multicolumn{1}{c}{} &  &  &  & \multicolumn{1}{c}{} &  &  &  &  \\
	De-biased Lasso (FF3) & 0.0062 & 0.0265 & 0.2352 & \multicolumn{1}{c}{0.9142} & 0.0134 & 0.1245 & 0.1077 & \multicolumn{1}{c}{2.2245} & 0.0120 & 0.1149 & 0.1046 & \multicolumn{1}{c}{1.6208} & -0.0236 & 0.6058 & -0.0390 & 1.0482 \\
	Lasso (FF5) & 0.0001 & 0.0042 & 0.0310 & \multicolumn{1}{c}{0.0861} &  &  &  & \multicolumn{1}{c}{} &  &  &  & \multicolumn{1}{c}{} &  &  &  &  \\
	De-biased Lasso (FF5) & 0.0062 & 0.0300 & 0.2124 & 0.9507 & 0.0184 & 0.1122 & 0.1639 & 2.2542 & 0.0076 & 0.0693 & 0.1098 & 1.6033 & 0.0083 & 0.0980 & 0.0844 & 1.0944 \\ \bottomrule
			\end{tabular}%
		}
	\end{table}
\end{landscape}

%%%%%%%%%%%%%%%%%%%%%%%%%%%%%%%%%%%%%%%%%%%%%%%%%%%%%%%%%%%%%%%%%%%%%%%%%%%%%%%%%%%%%%%%%%%%%%%%%%%%%%%%%%%%%%%%%%%%%%%%%%%%%
\begin{landscape}
	\begin{table}[]
		\centering
		\caption{\textit{Cumulative excess return (CER) and risk of non-sparse portfolios using monthly data.}
		}
		\label{tabtry}
		\resizebox{0.8\textwidth}{!}{%
			\begin{tabular}{@{}m{4cm}m{2cm}m{2cm}m{2cm}m{2cm}m{2cm}m{2cm}m{2cm}@{}}
			\toprule
			& \multicolumn{2}{c}{\begin{tabular}[c]{@{}c@{}}Asian \& Rus. \\ Fin. Crisis\\  (1997-1998)\end{tabular}} & \multicolumn{2}{c}{\begin{tabular}[c]{@{}c@{}}Argen. Great Depr. \\ \& dot-com bubble \\ (1999-2002)\end{tabular}} & \multicolumn{2}{c}{\begin{tabular}[c]{@{}c@{}}Fin. Crisis \\ (2007-2009)\end{tabular}}\\ \midrule
			& \textbf{CER} & \textbf{Risk} & \textbf{CER} & \textbf{Risk} & \textbf{CER} & \textbf{Risk}\\ \midrule
			EW & 0.2712 & \multicolumn{1}{c}{0.0547} & -0.0322 & \multicolumn{1}{c}{0.0519} &-0.4987& \multicolumn{1}{c}{0.1203} \\
			Index & 0.3222 & \multicolumn{1}{c}{0.0508} & -0.1698 & \multicolumn{1}{c}{0.0539}& -0.4924 & \multicolumn{1}{c}{0.0929}\\ \midrule
			& \multicolumn{6}{c}{\textbf{Markowitz Risk-Constrained (MRC)}} \\ \midrule
			MB & 2.1662 & \multicolumn{1}{c}{0.3381} & -0.1140 & \multicolumn{1}{c}{0.2916} & -3.0688 & \multicolumn{1}{c}{0.5101} \\
			CLIME & 1.3285 & \multicolumn{1}{c}{0.0892} & 0.4241 & \multicolumn{1}{c}{0.1297} & -3.0470 & \multicolumn{1}{c}{0.4735} \\
			LW & 0.9134 & \multicolumn{1}{c}{0.1021} & 0.3677 & \multicolumn{1}{c}{0.1412} & -0.3196 & \multicolumn{1}{c}{0.2751} \\
			FMB (PC) & 1.3153 & \multicolumn{1}{c}{0.0883} & 0.5016 & \multicolumn{1}{c}{0.1286} & -0.1312 & \multicolumn{1}{c}{0.1219}  \\
			FMB (FF1) & 2.0379 & \multicolumn{1}{c}{0.3029} & 0.0861 & \multicolumn{1}{c}{0.2660} & -2.7247 & \multicolumn{1}{c}{0.4301} \\ 	\midrule
			& \multicolumn{6}{c}{\textbf{Global Minimum-Variance Portfolio (GMV)}} \\ \midrule
			MB &  0.2791 & \multicolumn{1}{c}{0.0496} & -0.0470 & \multicolumn{1}{c}{0.0476} & -0.4637 & \multicolumn{1}{c}{0.1015} \\
			CLIME & 0.3960 & \multicolumn{1}{c}{0.0374} & -0.1224 & \multicolumn{1}{c}{0.0510} & -0.4588 & \multicolumn{1}{c}{0.0987} \\
			LW & 0.3127 & \multicolumn{1}{c}{0.0415} & -0.0952 & \multicolumn{1}{c}{0.0483} & -0.4013 & \multicolumn{1}{c}{0.0693} \\
			FMB (PC) & 0.4117 & \multicolumn{1}{c}{0.0364} & -0.1227 & \multicolumn{1}{c}{0.0505} & -0.3444 & \multicolumn{1}{c}{0.0393} \\
			FMB (FF1) & 0.2784 & \multicolumn{1}{c}{0.0487} & -0.0396 & \multicolumn{1}{c}{0.0468} & -0.4570 & \multicolumn{1}{c}{0.0986}\\	\bottomrule
		\end{tabular}%
		}
		
	\end{table}
\end{landscape}

\begin{landscape}
	\begin{table}[]
		\centering
		\caption{\textit{Cumulative excess return (CER) and risk of sparse portfolios using monthly data.}
		}
		\label{tabtrysparse}
		\resizebox{0.8\textwidth}{!}{%
			\begin{tabular}{@{}m{3cm}m{2cm}m{2cm}m{2cm}m{2cm}m{2cm}m{2cm}m{2cm}@{}}
			\toprule
			& \multicolumn{2}{c}{\begin{tabular}[c]{@{}c@{}}Asian \& Rus. \\ Fin. Crisis\\  (1997-1998)\end{tabular}} & \multicolumn{2}{c}{\begin{tabular}[c]{@{}c@{}}Argen. Great Depr. \\ \& dot-com bubble \\ (1999-2002)\end{tabular}} & \multicolumn{2}{c}{\begin{tabular}[c]{@{}c@{}}Fin. Crisis \\ (2007-2009)\end{tabular}} \\ \midrule
			& \textbf{CER} & \textbf{Risk} & \textbf{CER} & \textbf{Risk} & \textbf{CER} & \textbf{Risk}\\ \midrule	
			EW & 0.2712 & \multicolumn{1}{c}{0.0547} & -0.0322 & \multicolumn{1}{c}{0.0519} &-0.4987& \multicolumn{1}{c}{0.1203} \\
			Index & 0.3222 & \multicolumn{1}{c}{0.0508}& -0.1698 & \multicolumn{1}{c}{0.0539}& -0.4924 & \multicolumn{1}{c}{0.0929}\\ \midrule
			& \multicolumn{6}{c}{\textbf{Debiased MRC}} \\ \midrule	
			DL(PC) &0.2962 & \multicolumn{1}{c}{0.0261} & 0.1567 & \multicolumn{1}{c}{0.0217} & 0.1129 & \multicolumn{1}{c}{0.0408} \\
			DL(FF1) & 0.4149 & \multicolumn{1}{c}{0.0277} & 0.1681 & \multicolumn{1}{c}{0.0240} & -0.0258 & \multicolumn{1}{c}{0.0230} \\
			DL(FF3) & 0.2123 & \multicolumn{1}{c}{0.0142} & 0.1782 & \multicolumn{1}{c}{0.0186} & -0.0406 & \multicolumn{1}{c}{0.0202} \\ \midrule
			& \multicolumn{6}{c}{\textbf{Post-Lasso MRC}} \\ \midrule
			PL(PC) & 3.0881 & \multicolumn{1}{c}{0.2211} & 1.7153 & \multicolumn{1}{c}{0.1281} & 2.6131 & \multicolumn{1}{c}{0.1862}  \\
			PL(FF1) & 2.3433 & \multicolumn{1}{c}{0.1568} & 1.4470 & \multicolumn{1}{c}{0.1828} & 2.8639 & \multicolumn{1}{c}{0.2404}  \\
			PL(FF3) & 0.6691 & \multicolumn{1}{c}{0.1887} & -0.1561 & \multicolumn{1}{c}{0.1799} & -0.9998 & \multicolumn{1}{c}{0.1410}\\ \midrule
			& \multicolumn{6}{c}{\textbf{Post-Lasso GMV}} \\ \midrule				
			PL(PC) & 0.4403 & \multicolumn{1}{c}{0.0593} & 0.8150 & \multicolumn{1}{c}{0.0955} & -0.3694 & \multicolumn{1}{c}{0.1243}  \\
			PL(FF1) & 0.3385 & \multicolumn{1}{c}{0.0616} & 0.8151 & \multicolumn{1}{c}{0.0877} & -0.5545 & \multicolumn{1}{c}{0.1213}  \\
			PL(FF3) & 0.0711 & \multicolumn{1}{c}{0.0713} & 0.1458 & \multicolumn{1}{c}{0.1061} & 0.0295 & \multicolumn{1}{c}{0.0694}  \\ \bottomrule
		\end{tabular}%
		}
		
	\end{table}
\end{landscape}

%%%%%%%%%%%%%%%%%%%%%%%%%%%%%%%%%%%%%%%%%%%%%%%%%%%%%%%%%%%%%%%%%%%%%%%%%%%%%%%%%%%%%%%%%%%%%%%%%%%%%%%%%%%%%%%%%%%%%%%%%%%%%
\phantomsection
\renewcommand{\appendixpagename}{Appendix}
\renewcommand\appendixtocname{Appendix}
\begin{appendices}
	\renewcommand{\thesubsection}{A.\arabic{subsection}}
	\renewcommand{\theequation}{A.\arabic{equation}}
In this Appendix we collected the proofs of Theorems 2-5.
\subsection{Proof of Theorem \ref{theor2}}
The first part of Theorem \ref{theor2} was proved in \cite{fan2018elliptical} (see their proof of Theorem 2.1) under the assumptions \ref{A1}-\ref{A3}, \ref{B1}-\ref{B3} and $\log p=o(T)$. To prove the convergence rates for the precision matrix of the factor-adjusted returns, we follow \cite{CHANG2018}, \cite{Caner2018} and \cite{Caner2019}. Using the facts that $\max_{i\leq p} (1/T)\sum_{t=1}^{T} \abs{\hat{\varepsilon}_{it}-\varepsilon_{it}}=\mathcal{O}_P(\omega_{T}^2)$ and $\max_{i,t}\abs{\hat{\varepsilon}_{it}-\varepsilon_{it}}=\mathcal{O}_P(\omega_{T})=o_P(1)$, we get
\begin{align}\label{Ap1}
\max_{1\leq j\leq p} \norm{\widehat{\bgamma}_j-\bgamma_j}_1=\mathcal{O}_P(\bar{s}\omega_{T}),
\end{align}
where $\widehat{\bgamma}_j$ was defined in \eqref{e21}. The proof of \ref{Ap1} is similar to the proof of the equation (23) of \cite{CHANG2018}, with $\omega_{T}=\sqrt{\log p/T}$ for their case. Similarly to \cite{Caner2019}, consider the following linear model:
\begin{align} \label{Ap2}
&\widehat{\bvarepsilon}_{j}=\widehat{\bE}_{-j}\bgamma_j+\boldeta_j, \ \text{for} \ j=1,\ldots,p,\\
&\E{\boldeta'_{j}\widehat{\bE}_{-j}} = 0.\nonumber
\end{align}
\cite{Buhlmann2014} and \cite{CHANG2018} showed that
\begin{align}\label{Ap3}
\max_{1\leq j\leq p} \norm{\boldeta'_{j}\widehat{\bE}_{-j}}_{\infty}/T=\mathcal{O}_P(\omega_{T}).
\end{align}
Let $\tau_{j}^{2}\defeq \E{\boldeta'_{j}\boldeta_{j}}$, then we have
\begin{align}\label{Ap4}
\max_{1\leq j\leq p} \norm{ \boldeta'_{j}\boldeta_{j}/T-\tau_{j}^{2} } = \mathcal{O}_P(\omega_{T}).
\end{align}

Note that the rate in \eqref{Ap4} is the same as in Lemma 1 of \cite{CHANG2018} with $\omega_{T}=\sqrt{\log p/T}$ for their case. However, the rate in \eqref{Ap4} is different from the one derived in \cite{Buhlmann2014} since we allow time-dependence between factor-adjusted returns.\\

Recall that $\hat{\tau}_{j}^{2}=\norm{\widehat{\bvarepsilon}_{j}-\widehat{\bE}_{-j}\widehat{\bgamma}_j}_{2}^{2}/T+\lambda_j\norm{\widehat{\bgamma}_j}_1$. Using triangle inequality, we have:
\begin{align*}
\max_{1\leq j\leq p} \abs{\hat{\tau}_{j}^{2}-\tau_{j}^{2}} &\leq \underbrace{ \max_{1\leq j\leq p} \abs{\boldeta'_{j}\boldeta_{j}/T-\tau_{j}^{2}} }_{\text{\rom{1}}} + \underbrace{ \max_{1\leq j\leq p} \abs{\boldeta'_{j}\widehat{\bE}_{-j}(\widehat{\bgamma}_j-\bgamma_j)/T  } }_{\text{\rom{2}}}\\
&+ \underbrace{ \max_{1\leq j\leq p} \abs{\boldeta'_{j}\widehat{\bE}_{-j}\bgamma_j/T} }_{\text{\rom{3}}} + \underbrace{ \max_{1\leq j\leq p} \bgamma'_{j} \widehat{\bE}'_{-j}\widehat{\bE}_{-j}(\widehat{\bgamma}_j-\bgamma_j)/T}_{\text{\rom{4}}}.
\end{align*}

The first term was bounded in \ref{Ap4}, we now bound the remaining terms:
\begin{align*}
\text{\rom{2}}\leq \max_{1\leq j\leq p} \norm{\boldeta'_{j}\widehat{\bE}_{-j}/T}_{\infty} \max_{1\leq j\leq p} \norm{\widehat{\bgamma}_j-\bgamma_j}_1=\mathcal{O}_P(\bar{s}\omega_{T}^{2}),
\end{align*}
where we used \ref{Ap1} and \ref{Ap3}. For \rom{3} we have
\begin{align*}
\text{\rom{3}}\leq \max_{1\leq j\leq p} \norm{\boldeta'_{j}\widehat{\bE}_{-j}/T}_{\infty} \max_{1\leq j\leq p} \norm{\bgamma_j}_1=\mathcal{O}_P(\sqrt{\bar{s}}\omega_{T}),
\end{align*}
where we used \ref{Ap3} and the fact that $\norm{\bgamma_j}_1\leq \sqrt{s_j}\norm{\bgamma_j}_2=\mathcal{O}(\sqrt{s_j})$.
To bound the last term, we use KKT conditions in node-wise regression:
\begin{align*}
\max_{1\leq j\leq p} \norm{ \widehat{\bE}'_{-j}\widehat{\bE}_{-j}(\widehat{\bgamma}_j-\bgamma_j)/T}_{\infty} \leq \max_{1\leq j\leq p} \norm{\widehat{\bE}'_{-j}\boldeta_{j}/T}_{\infty} + \max_{1\leq j\leq p} \lambda_j = \mathcal{O}_P(\omega_{T}),
\end{align*}
where we used \ref{Ap3} and $\lambda_j \asymp \omega_{T}$. It follows that
\begin{align*}
\text{\rom{4}} =  \mathcal{O}_P(\omega_{T})  \max_{1\leq j\leq p} \norm{\bgamma_j}_1 = \mathcal{O}_P(\sqrt{\bar{s}}\omega_{T}).
\end{align*}
Therefore, we now have shown that
\begin{align}
\max_{1\leq j\leq p} \abs{\hat{\tau}_{j}^{2}-\tau_{j}^{2}} = \mathcal{O}_P(\sqrt{\bar{s}}\omega_{T}).
\end{align}
Using the fact that $1/\tau_{j}^{2}=\mathcal{O}(1)$, we also have
\begin{align}
1/\hat{\tau}_{j}^{2}-1/\tau_{j}^{2} = \mathcal{O}_P(\sqrt{\bar{s}}\omega_{T}).
\end{align}
Finally, using the analysis in (B.51)-(B.53) of \cite{Caner2018}, we get
\begin{align}
\max_{1\leq j \leq p} \norm{\widehat{\bTheta}_{\varepsilon,j}-\bTheta_{\varepsilon,j}}_1=\mathcal{O}_P(s_T\omega_{T}).
\end{align}
To prove the second rate for the precision of the factor-adjusted returns, we note that
\begin{align}\label{Ap8}
\max_{1\leq j\leq p} \norm{\widehat{\bgamma}_j-\bgamma_j}_2=\mathcal{O}_P(\sqrt{\bar{s}}\omega_{T}),
\end{align}
which was obtained in \cite{CHANG2018} (see their Lemma 2). We can write
\begin{align}
\max_{1\leq j \leq p} \norm{\widehat{\bTheta}_{\varepsilon,j}-\bTheta_{\varepsilon,j}}_2\leq \max_{1\leq j \leq p} \lbrack \norm{\widehat{\bgamma}_j-\bgamma_j}_2/\hat{\tau}_{j}^{2} + \norm{\bgamma_j}_2  1/\hat{\tau}_{j}^{2}-1/\tau_{j}^{2}   \rbrack = \mathcal{O}_P(\sqrt{\bar{s}}\omega_{T}).
\end{align}

\subsection{Proof of Theorem \ref{theor3}}
Let $\widehat{\bJ}=\widehat{\bLambda}^{1/2}\widehat{\bGamma}'\widehat{\bTheta}_{\varepsilon}\widehat{\bGamma}\widehat{\bLambda}^{1/2}$ and $\widetilde{\bJ}=\widetilde{\bLambda}^{1/2}\widetilde{\bGamma}'\bTheta_{\varepsilon}\widetilde{\bGamma}\widetilde{\bLambda}^{1/2}$. Also, define 
\begin{equation*}
\Delta_{\text{inv}} = \widehat{\bTheta}_{\varepsilon}\widehat{\bGamma}\widehat{\bLambda}^{1/2}(\bI_{K}+\widehat{\bJ})^{-1}\widehat{\bLambda}^{1/2}\widehat{\bGamma}'\widehat{\bTheta}_{\varepsilon}-\bTheta_{\varepsilon}\widetilde{\bGamma}\widetilde{\bLambda}^{1/2} (\bI_{K}+\widetilde{\bJ})^{-1} \widetilde{\bLambda}^{1/2} \widetilde{\bGamma}' \bTheta_{\varepsilon}.
\end{equation*}
Using Sherman-Morrison-Woodbury formulas in \ref{equa18}, we have
\begin{align} \label{Ap10}
\vertiii{\widehat{\bTheta}-\bTheta}_1&\leq \vertiii{\widehat{\bTheta}_{\varepsilon}-\bTheta_{\varepsilon}}_1+\vertiii{\Delta_{\text{inv}}}_1.
\end{align}
As pointed out by \cite{fan2018elliptical}, $\vertiii{\Delta_{\text{inv}}}_1$ can be bounded by the following three terms:
\begin{align*}
&\vertiii{(\widehat{\bTheta}_{\varepsilon}-\bTheta_{\varepsilon}) \widetilde{\bGamma}\widetilde{\bLambda}^{1/2}(\bI_{K}+\widetilde{\bJ})^{-1} \widetilde{\bLambda}^{1/2} \widetilde{\bGamma}'\bTheta_{\varepsilon}}_1 = \mathcal{O}_P(\bar{s}\omega_{T}\cdot p\cdot \frac{1}{p}\cdot \sqrt{\bar{s}}),\\
&\vertiii{\bTheta_{\varepsilon}(\widehat{\bGamma}\widehat{\bLambda}^{1/2} - \widetilde{\bGamma}\widetilde{\bLambda}^{1/2} ) (\bI_{K}+\widetilde{\bJ})^{-1} \widetilde{\bLambda}^{1/2} \widetilde{\bGamma}'\bTheta_{\varepsilon}}_1 = \mathcal{O}_P(\sqrt{\bar{s}}\cdot p\omega_{T}\cdot \frac{1}{p}\cdot \sqrt{\bar{s}}),\\
& \vertiii{\bTheta_{\varepsilon}\widetilde{\bLambda}^{1/2} \widetilde{\bGamma}'((\bI_{K}+\widehat{\bJ})^{-1} - (\bI_{K}+\widetilde{\bJ})^{-1})\widetilde{\bGamma}'\bTheta_{\varepsilon} }_1 = \mathcal{O}_P(\sqrt{\bar{s}}\cdot \frac{1}{p}\cdot p\bar{s}\omega_{T}\sqrt{\bar{s}}  ).
\end{align*}
To derive the above rates we used \ref{B1}-\ref{B3}, Theorem \ref{theor2} and the fact that $\norm{\widehat{\bGamma}\widehat{\bLambda}\widehat{\bGamma}'-\bB\bB'}_{\text{F}}=\mathcal{O}_P(p\omega_{T})$. The second rate in Theorem \ref{theor3} can be easily obtained using the technique described above for the $l_2$-norm.

\subsection{Lemmas for Theorems \ref{theor3a}-\ref{theor4}}
\begin{lem}\label{theorA1caner}
	Under the assumptions of Theorem \ref{theor3},
	\begin{enumerate}[label=(\alph*)]
		\item $\norm{\widehat{\boldm}-\boldm}_{\text{max}}=\mathcal{O}_P(\sqrt{\log p/T})$, where $\boldm$ is the unconditional mean of stock returns defined in Subsection 3.2, and $\widehat{\boldm}$ is the sample mean.
		\item $\vertiii{\bTheta}_1=\mathcal{O}(\bar{s})$.	
	\end{enumerate}
\end{lem}
\begin{proof}
	~
	\begin{enumerate}[label=(\alph*)]
		\item The proof of Part (a) is provided in \cite{CHANG2018} (Lemma 1).
		\item To prove Part (b) we use Sherman-Morrison-Woodbury formula in \ref{equa18}:
		\begin{align}\label{A6new}
		\vertiii{\bTheta}_1 &\leq \vertiii{\bTheta_{\varepsilon}}_1+\vertiii{\bTheta_{\varepsilon}\bB\lbrack\bI_{K}+\bB'\bTheta_{\varepsilon}\bB \rbrack^{-1}\bB'\bTheta_{\varepsilon}}_1 \nonumber\\
		&=\mathcal{O}(\sqrt{\bar{s}})+\mathcal{O}(\sqrt{\bar{s}}\cdot p \cdot \frac{1}{p}  \cdot \sqrt{\bar{s}}) = \mathcal{O}(\bar{s}).
		\end{align}
		The last equality in \ref{A6new} is obtained under the assumptions of Theorem \ref{theor4}. This result is important in several aspects: it shows that the sparsity of the precision matrix of stock returns is controlled by the sparsity in the precision of the idiosyncratic returns. Hence, one does not need to impose an unrealistic sparsity assumption on the precision of returns a priori when the latter follow a factor structure - sparsity of the precision once the common movements have been taken into account would suffice.
	\end{enumerate}
\end{proof}
\begin{lem}\label{lemA1caner}
	Define $\theta = \boldm'\bTheta\boldm/p$ and $g = \sqrt{\boldm'\bTheta\boldm}/p$. Also, let $\widehat{\theta}=\widehat{\boldm}'\widehat{\bTheta}\widehat{\boldm}/p$ and $\widehat{g}=\sqrt{\widehat{\boldm}'\widehat{\bTheta}\widehat{\boldm}}/p$. Under the assumptions of Theorem \ref{theor3}:
	\begin{enumerate}[label=(\alph*)]
		\item $\theta=\mathcal{O}(1)$.
		\item  $\abs{\widehat{\theta}-\theta}=\mathcal{O}_P(\bar{s}^2\omega_{T} )=o_P(1)$.
		\item $\abs{\widehat{y}-y}=\mathcal{O}_P(\bar{s}^2\omega_{T} )=o_P(1)$, where $y$ was defined in \eqref{e2.17}.
%		\item  $\abs{\widehat{d}-d}=\mathcal{O}_P(\bar{s}^2\omega_{T} )=o_P(1)$.
		\item $\abs{\widehat{g}-g}=\mathcal{O}_P\Big(\lbrack\bar{s}^2\omega_{T}\rbrack^{1/2}\Big)=o_P(1)$.
	\end{enumerate}
\end{lem}
\begin{proof}
	~
	\begin{enumerate}[label=(\alph*)]
		\item Part (a) is trivial and follows directly from $\vertiii{\bTheta}_2=\mathcal{O}(1)$.
		\item First, rewrite the expression of interest:
		\begin{align} \label{A6}
		\widehat{\theta}-\theta &= \lbrack (\widehat{\boldm}-\boldm)'(\widehat{\bTheta}-\bTheta) (\widehat{\boldm}-\boldm)  \rbrack/p + \lbrack (\widehat{\boldm}-\boldm)'\bTheta (\widehat{\boldm}-\boldm)  \rbrack/p \nonumber\\
		&+\lbrack 2(\widehat{\boldm}-\boldm)'\bTheta\boldm  \rbrack/p + \lbrack 2 \boldm'(\widehat{\bTheta}-\bTheta) (\widehat{\boldm}-\boldm)  \rbrack/p \nonumber\\
		&+ \lbrack\boldm'(\widehat{\bTheta}-\bTheta) \boldm  \rbrack/p.
		\end{align}
		We now bound each of the terms in \ref{A6} using the expressions derived in \cite{Caner2019} (see their Proof of Lemma A.3), Lemma \ref{theorA1caner} and the fact that $\log p/T=o(1)$.
		\begin{align}\label{A7}
		\abs{ (\widehat{\boldm}-\boldm)'(\widehat{\bTheta}-\bTheta) (\widehat{\boldm}-\boldm)  }/p &\leq \norm{\widehat{\boldm}-\boldm }_{\text{max}}^2\vertiii{\widehat{\bTheta}-\bTheta }_1 \nonumber\\ &= \mathcal{O}_P\Big(\frac{\log p}{T} \cdot \bar{s}^2\omega_{T} \Big)
		\end{align}
		\begin{align}
		\abs{ (\widehat{\boldm}-\boldm)'\bTheta (\widehat{\boldm}-\boldm)  }/p \leq \norm{\widehat{\boldm}-\boldm }_{\text{max}}^2\vertiii{\bTheta}_1 = \mathcal{O}_P\Big(\frac{\log p}{T} \cdot \bar{s} \Big).
		\end{align}
		\begin{align}
		\abs{ (\widehat{\boldm}-\boldm)'\bTheta\boldm }/p \leq \norm{\widehat{\boldm}-\boldm }_{\text{max}}\vertiii{\bTheta}_1 = \mathcal{O}_P\Big(\sqrt{\frac{\log p}{T}} \cdot \bar{s} \Big).
		\end{align}
		\begin{align}
		\abs{\boldm'(\widehat{\bTheta}-\bTheta) (\widehat{\boldm}-\boldm)}/p &\leq \norm{\widehat{\boldm}-\boldm }_{\text{max}}\vertiii{\widehat{\bTheta}-\bTheta }_1 \nonumber\\ &= \mathcal{O}_P\Big(\sqrt{\frac{\log p}{T}} \cdot \bar{s}^2\omega_{T} \Big).
		\end{align}
		\begin{align}
		\abs{\boldm'(\widehat{\bTheta}-\bTheta) \boldm  }/p \leq \vertiii{\widehat{\bTheta}-\bTheta }_1 = \mathcal{O}_P\Big(\bar{s}^2\omega_{T} \Big).
		\end{align}
		\item Part (c) trivially follows from Part (b).
		\item This is a direct consequence of Part (b) and the fact that $\sqrt{\widehat{\theta}-\theta}\geq \sqrt{\widehat{\theta}}-\sqrt{\theta}$.
	\end{enumerate}	
\end{proof}
\subsection{Proof of Theorem \ref{theor3a}}
Using the definition of MRC weight in \eqref{e158}, we can rewrite
\begin{align*}
&\norm{\widehat{\bw}_{\text{MRC}}-\bw_{\text{MRC}}}_1 \leq \frac{\frac{g}{p}\Big[\norm{ (\widehat{\bTheta}-\bTheta) (\widehat{\boldm}-\boldm) }_1+  \norm{ (\widehat{\bTheta}-\bTheta) \boldm }_1 + \norm{ \bTheta (\widehat{\boldm}-\boldm) }_1  \Big] + \abs{\widehat{g}-g}\norm{\bTheta\boldm}_1}{\abs{ \widehat{g}}g}\\
&\leq \frac{\frac{g}{p}\Big[p\vertiii{ \widehat{\bTheta}-\bTheta }_1\norm{ (\widehat{\boldm}-\boldm) }_{\text{max}}+ p\vertiii{ \widehat{\bTheta}-\bTheta }_1 \norm{ \boldm }_{\text{max}} + p\vertiii{  \bTheta }_1\norm{ (\widehat{\boldm}-\boldm) }_{\text{max}}  \Big] + p\abs{\widehat{g}-g}\vertiii{ \bTheta }_1\norm{\boldm}_{\text{max}}}{\abs{ \widehat{g}}g}\\
&=\mathcal{O}_P \Big( \bar{s}^2\omega_{T} \cdot \sqrt{\frac{\log p}{T}}  \Big) +  \mathcal{O}_P \Big(\bar{s}^2\omega_{T} \Big)+ \mathcal{O}_P\Big(\bar{s}\cdot \sqrt{\frac{\log p}{T}} \Big) +  \mathcal{O}_P \Big(\lbrack \bar{s}^2\omega_{T} \rbrack^{1/2}\cdot \bar{s} \Big) = o_P(1),
\end{align*}
where we used Lemmas 1-2 to obtain the rates.
\subsection{Proof of Theorem \ref{theor4}}
The KKT conditions for the nodewise Lasso in \eqref{e21} imply that
\begin{align*}
\hat{\tau}_{j}^{2}=(\widehat{\bvarepsilon}_{j}-\widehat{\bE}_{-j}\widehat{\bgamma}_{j})'\widehat{\bvarepsilon}_{j}/T, \ \text{hence,} \quad \widehat{\bvarepsilon}'_{j}\widehat{\bE}\widehat{\bTheta}'_{\varepsilon,j}/T=1.
\end{align*}
As shown in \cite{Buhlmann2014}, these KKT conditions also imply that
\begin{align}
\norm{\widehat{\bE}'_{-j}\widehat{\bE}\widehat{\bTheta}_{\varepsilon,j}}_{\infty}/T \leq \lambda_j/\hat{\tau}_{j}^{2}.
\end{align}
Therefore, the estimator of precision matrix needs to satisfy the following \enquote{extended KKT} condition:
\begin{align}\label{Ap19}
\norm{\widehat{\bSigma}_{\varepsilon}\widehat{\bTheta}'_{\varepsilon,j}-\be_j}_{\infty} \leq \lambda_j/\hat{\tau}_{j}^{2},
\end{align}
where $\be_j$ is the $j$-th unit column vector. Combining the rate in $\ell_1$ norm in Theorem \ref{theor3} and \eqref{Ap19}, we have:
\begin{align}\label{Ap19B}
\norm{\widehat{\bSigma}\widehat{\bTheta}'_{j}-\be_j}_{\infty} \leq \lambda_j/\hat{\tau}_{j}^{2},
\end{align}
Using the definition of $\Delta$ in \eqref{ee73}, it is straightforward to see that
\begin{align}\label{Ap20}
\norm{\Delta}_{\infty}/\sqrt{T} = \norm{ (\widehat{\bTheta}\widehat{\bSigma}-\bI_p)(\widehat{\bw}-\bw)}_{\infty}\leq \norm{ \widehat{\bTheta}\widehat{\bSigma} - \bI_p }_{\infty} \norm{\widehat{\bw}-\bw}_1.
\end{align}
Therefore, combining \eqref{Ap19B} and \eqref{Ap20}, we have 
\begin{align}
\norm{\Delta}_{\infty} \leq \sqrt{T} \norm{\widehat{\bw}-\bw}_1\max_{j}\lambda_j/\hat{\tau}_{j}^{2}&=\mathcal{O}_P\Big(\sqrt{T}\cdot (s_0\vee \bar{s}^2)\omega_{T}\cdot \omega_{T}\Big)\\ &= \mathcal{O}_P\Big((s_0\vee \bar{s}^2)\Big(\log p/\sqrt{T}+\sqrt{T}/p\Big)\Big)=o_P(1).
\end{align}
Finally, we show that  $\norm{\widehat{\bOmega}-\bTheta}_{\infty}=o_P(1)$. Using Theorem \ref{theor3} and Lemma \ref{theorA1caner} we have $\norm{\widehat{\bTheta}_{j}}_1=\mathcal{O}_P(s_j)$. Also, 
\begin{align}\label{Ap22}
\widehat{\bOmega}=\widehat{\bTheta}\widehat{\bSigma}\widehat{\bTheta}' = (\widehat{\bTheta}\widehat{\bSigma}-\bI_p)\widehat{\bTheta}'+\widehat{\bTheta}'.
\end{align}
And using \ref{Ap19B} and \ref{Ap20} together with $\max_j\lambda_js_{j}^{2}=o_P(1)$:
\begin{align}\label{Ap23}
\norm{(\widehat{\bTheta}\widehat{\bSigma}-\bI_p)\widehat{\bTheta}'}_{\infty}\leq \max_j\lambda_j \norm{\widehat{\bTheta}_{j}}_1/\hat{\tau}_{j}^{2}=o_P(1).
\end{align}
It follows that
\begin{align}\label{Ap24}
\norm{\widehat{\bTheta}-\bTheta}_{\infty} \leq \max_j\norm{\widehat{\bTheta}_j-\bTheta_j}_2\leq \max_j\lambda_j\sqrt{s_j}=o_P(1).
\end{align}
Combining \ref{Ap22}-\ref{Ap24} completes the proof.
\cleardoublepage
\end{appendices}
\cleardoublepage
\end{document}